\documentclass[a4paper,thm-restate]{lipics-v2021}
\usepackage[utf8]{inputenc}
\usepackage{babel}
\usepackage{amsmath}
\usepackage{graphicx}
\usepackage{verbatim}
\usepackage{varioref}
\usepackage{refstyle}
\usepackage{color}
\usepackage{pdfpages}
\usepackage{float}
\usepackage{amsthm}
\usepackage{enumitem}
\usepackage{amssymb}
\usepackage{thm-restate}

\newcommand{\eps}{\varepsilon}
\newcommand{\poly}{\operatorname{poly}}

\makeatletter
\def\@endtheorem{\endtrivlist}
\makeatother

\newtheorem{invariant}[theorem]{Invariant}
\newtheorem*{corollaryS}{Corollary}
\newtheorem*{theoremS}{Theorem}

\newtheorem*{lemmaS}{Lemma}

\usepackage{hyperref}
\usepackage{cleveref}

\usepackage{algorithm}
\usepackage[noend]{algpseudocode}

\hideLIPIcs
\nolinenumbers


\title{Fully-dynamic $\alpha + 2$ Arboricity Decompositions and Implicit Colouring}

\author{Aleksander B. G. Christiansen}{Technical University of Denmark, Lyngby, Denmark}{abgch@dtu.dk}{no orcID}{Partially supported by the VILLUM Foundation grant 37507 ``Efficient Recomputations for Changeful Problems''.}
\author{Eva Rotenberg}{Technical University of Denmark, Lyngby, Denmark}{erot@dtu.dk}{http://orcid.org/0000-0001-5853-7909}{Partially supported by Independent Research Fund Denmark grants 2020-2023 (9131-00044B) ``Dynamic Network Analysis'' and 2018-2021 (8021-00249B) ``AlgoGraph'', and the VILLUM Foundation grant 37507 ``Efficient Recomputations for Changeful Problems''.}

\authorrunning{Aleksander B. G. Christiansen and Eva Rotenberg} 
\Copyright{Aleksander B. G. Christiansen and Eva Rotenberg}


\category{}

\date{1. October 2021}

\begin{document}

\maketitle
\begin{abstract}
	In the implicit dynamic colouring problem, the task is to maintain a representation of a proper colouring as a dynamic graph is subject to insertions and deletions of edges, while facilitating interspersed queries to the colours of vertices. The goal is to use few colours, while still efficiently handling edge-updates and responding to colour-queries. 
For an $n$-vertex dynamic graph of arboricity 
$\alpha$, we present an algorithm that maintains an implicit vertex colouring with $4\cdot 2^{\alpha}$ colours, in amortised poly-$\log n$ update time, and with $O(\alpha \log n)$ worst-case query time. 
The previous best implicit dynamic colouring algorithm uses $2^{40\alpha}$ colours, and has a more efficient update time of $O(\log^3 n)$ and the same query time of $O(\alpha \log n)$~\cite{henzinger2020explicit}. 

For graphs undergoing arboricity $\alpha$ preserving updates, we give a fully-dynamic 
$\alpha+2$ arboricity decomposition in 
$\poly(\log n,\alpha)$ 
time, which matches the 
number of forests in the best near-linear static algorithm by Blumenstock and Fischer~\cite{blumenstock2019constructive} who obtain $\alpha+2$ forests in near-linear time. 

Our construction goes via dynamic bounded out-degree orientations, where we present 
a fully-dynamic explicit, deterministic, worst-case algorithm for $\lfloor (1+\varepsilon)\alpha \rfloor + 2$ bounded out-degree orientation with update time $O(\varepsilon^{-6}\alpha^2 \log^3 n)$.
The state-of-the-art explicit, deterministic, worst-case algorithm for bounded out-degree orientations maintains a $\beta\cdot \alpha + \log_{\beta} n$ out-orientation in $O(\beta^2\alpha^2+\beta\alpha\log_{\beta} n)$ time~\cite{kopelowitz2013orienting}. 

\end{abstract}

\setcounter{secnumdepth}{2}
\section{Introduction}
Graph colouring is a well-studied problem in computer science and discrete mathematics and has many applications such as planar routing and network optimization \cite{RealLifeCol}. 
A proper colouring of a graph $G = (V,E)$ on $n$ vertices is an assignment of colours to each vertex in $V(G)$ such that no neighbours receive the same colour. 
We are interested in minimising the number of colours used. 
The minimum number of colours that can be used to properly colour $G$, is called the \emph{chromatic number} of $G$. 
It is NP-hard to even approximate the chromatic number to within a factor of $n^{1-\eps}$ for all $\eps > 0$ \cite{10.1145/1132516.1132612,DBLP:conf/icalp/KhotP06}, but colourings with respect to certain parameters can be efficiently computed.
For instance, it is well known that if a graph is uniformly sparse in the sense that we can decompose it into $k$ forests, then we can efficiently compute a colouring: the sparsity of the graph ensures that every subgraph has a vertex of degree at most $2k-1$,
allowing us to compute a $2k$ colouring of the graph in linear time by colouring the vertices in a clever order.
The minimum number of forests that the graph can be decomposed into is called the \emph{arboricity} of $G$.
In the past decades, much work has gone into the study of \textit{dynamic algorithms} that are able to efficiently update a solution, as the problem undergoes updates. A general question about dynamic problems is: which (near-) linear-time solvable problems have polylogarithmic updatable solutions? 

We study the problem of maintaining a proper colouring of a dynamic graph with bounded arboricity. 
This class of graphs encompasses, for instance, dynamic planar graphs where $\alpha \leq 3$.
Here, the graph undergoes changes in the form of insertions and deletions of edges and one needs to maintain a proper colouring of the vertices with fast update times. 
We distinguish between two scenarios: one where, as is the case for dynamic planar graphs, we have access to an upper bound on the arboricity $\alpha_{max}$ throughout all updates, and one where we do not. 
Note that due to insights presented in~\cite{DBLP:conf/stoc/SawlaniW20}, we can often turn an algorithm for the first scenario into an algorithm for the second by scheduling updates to $O(\log{n})$ (partial) copies of the graph, thus incurring only an $O(\log n)$ overhead in the update time.

Barba et al.\ \cite{Barba} showed that one cannot hope to maintain a proper, explicit vertex-colouring of a dynamic forest with a constant number of colours in poly-logarithmic update time. 
Consequently, we cannot maintain explicit colourings where the number of colours depend entirely on $\alpha$ with poly-logarithmic update time - even if we know an upper bound on $\alpha$.
This motivated Henzinger et al.\ \cite{henzinger2020explicit} to initiate the study of \emph{implicit} colourings.
Here, instead of storing the colours of vertices explicitly in memory, a queryable data structure is provided which after some computations returns the colour of a vertex. 
If one queries the colours of two neighbouring vertices between updates, the returned colours must differ. 
Now, we can circumvent the lower bound by using known data structures for maintaining information in dynamic forest to 2-colour dynamic forests in poly-logarithmic update time.
Henzinger et al.\ \cite{henzinger2020explicit} use this to colour graphs via an \emph{arboricity decomposition} i.e.\ a decomposition of the graph into forests. 
They present a dynamic algorithm that maintains an implicit proper $2^{O(\alpha)}$-colouring of a dynamic graph $G$ with arboricity $\alpha$.
Their algorithm adapts to $\alpha$, but in return it hides a constant (around 40) in the asymptotic notation.
Even if one has an upper bound $\alpha_{max}$ on $\alpha$, the currently best obtainable colouring uses $2^{4(\alpha_{max}+1)}$ colours by combining 
the arboricity decomposition algorithm from Henzinger et al.\ \cite{henzinger2020explicit} with an algorithm of Brodal \& Fagerberg \cite{Brodal99dynamicrepresentations} that maintains a $2(\alpha_{max}+1)$ bounded out-degree orientation.
Both of these algorithms use a lot of colours. Even for planar graphs with arboricity at most 3, $2^{16} > 60.000$ colours are used. This is quite far from $4$ colours, which is always sufficient~\cite{AppelH76,RobertsonSST96}, or the 5 colouring that can be computed in linear time \cite{Matula1980TwoLA,ChibaNS81,Frederickson84}.

\subparagraph{Dynamic arboricity decompositions:}
Both colouring algorithms go via dynamic $\alpha'$-bounded out-orientations. Here, the goal is to orient the edges of the graph while keeping out-degrees low. These are then turned into dynamic $2\alpha'$-arboricity decompositions. 
By 2-colouring each forest, such a decomposition yields a $2^{2\alpha'}$ colouring.
Thus the lower $\alpha'$ is, the fewer colours we use.
There has been a lot of work on maintaining dynamic low out-orientations 
\cite{Brodal99dynamicrepresentations, berglinetal:LIPIcs:2017:8263,kopelowitz2013orienting,He2014OrientingDG,10.1145/3392724},  
and much of this work aim to improve update complexity by relaxing the allowed out-degree.
Motivated by implicit colourings, we provide a different trade-off, 
providing a lower $\alpha'$ value within $\operatorname{polylog}(n,\alpha)$ update time.
Specifically, a 
$\lfloor (1+\eps)\alpha \rfloor+2$ dynamic out-orientation with $O(\log^{3}(n)\alpha^2/\eps^{6})$ update-time adaptive to $\alpha$, and an $\alpha+2$ dynamic arboricity decomposition with $O(\operatorname{poly}(\log{n},\alpha_{max}))$ update time, when we have an upper bound $\alpha_{max}$ on the arboricity. 
Our algorithm maintaining the arboricity decomposition matches the number of forests obtained by the best static algorithm running in
near-linear time~\cite{blumenstock2019constructive}. 

These algorithms may also be interesting as they go below the $2\alpha$ barrier on out-edges and forests respectively. 
In the static case there exist simple and elegant algorithms computing $2\alpha-1$ out-orientations and arboricity decompositions in linear time \cite{DBLP:journals/dam/ArikatiMZ97,10.1016/0020-0190(94)90121-X}. 
For exact algorithms, the state-of-the-art algorithms spend time $O(m^{10/7})$ \cite{Madry} or $O(m\sqrt{n})$ \cite{LeeSidford} for the out-orientation problem, and $\tilde{O}(m^{3/2})$ for the arboricity decomposition problem \cite{10.1145/62212.62252,10.5555/313651.313667}. 
Even statically computing an $\alpha + 1$ out-orientation \cite{10.1007/11940128_56} resp.\  an $\alpha + 2$ arboricity decomposition \cite{blumenstock2019constructive} takes $\tilde{O}(m)$ time. 
In the dynamic case, the out-orientation with the lowest bound on the out-degree with $O(\poly(\log{n},\alpha))$ update time seem to be the algorithm of Brodal \& Fagerberg \cite{Brodal99dynamicrepresentations} that achieves $2(\alpha_{max}+1)$ out-degree. 
In \cite{Brodal99dynamicrepresentations}, it is also noted that determining exactly the complexity of maintaining a $d$ out-orientation for $d \in [\alpha,2\alpha]$ is a 'theoretically interesting direction for further research'. We make some progress in this direction by showing how to maintain a $\lfloor (1+\eps)\alpha\rfloor +2$ out-orientation with $\poly(\log{n},\alpha,\varepsilon^{-1})$ update time.
Thus, if $\alpha$ is a constant, we may carefully choose $\varepsilon$ to obtain a polylogarithmic $\alpha + 2$ out-orientation.

\subsection{Results}
Let $G$ be a dynamic graph with $n$ vertices undergoing insertion and deletions of edges, and let $\alpha$ be the current arboricity of the graph; that is $\alpha$ might change, when edges are inserted and deleted.
If we at all times have an upper bound $\alpha_{max}$ on $\alpha$, we say that $G$ is undergoing an $\alpha_{max}$ preserving sequence of updates. 
We have the following:
\begin{theorem}
\label{OutOrientationMain}
For $1> \eps > 0$, there exists a fully-dynamic algorithm maintaining an explicit $((1+\eps)\alpha+2)$-bounded out-degree orientation with worst-case insertion time $O(\log^{3}{n}\cdot{}\alpha^2/\eps^6)$ and worst-case deletion time $O(\log^{3}{n}\cdot{}\alpha/\eps^4)$
\end{theorem}
Using pseudoforest decompositions, we obtain a fully dynamic, implicit colouring algorithm:
\begin{corollary}
\label{cor:ColouringSecondary}
Given a dynamic graph with $n$ vertices, there exists a fully dynamic algorithm that maintains an implicit $2\cdot{}3^{(1+\varepsilon)\alpha}$ colouring with an amortized update time of $O(\log^{4}{n}\cdot{}\alpha^2/\eps^6)$ and a query time of $O(\alpha\log{n})$.
\end{corollary}
By moving edges between pseudoforests, we can turn the pseudoforest decomposition into a forest decomposition. This also gives a colouring algorithm using fewer colours.
\begin{theorem}
\label{ArbDecompMain}
Given an initially empty and dynamic graph undergoing an arboricity $\alpha_{max}$ preserving sequence of updates, there exists an algorithm maintaining a $\lfloor (1+\eps)\alpha \rfloor+2$ arboricity decomposition with an amortized update time of $O(\operatorname{poly}(\log{n},\alpha_{max},\varepsilon^{-1}))$. 
In particular, setting $\eps < \alpha_{max}^{-1}$ yields $\alpha+2$ forests with an amortized update time of $O(\operatorname{poly}(\log{n},\alpha_{max}))$. 
\end{theorem}
\begin{corollary}
\label{cor:ColouringMain}
Given a dynamic graph with $n$ vertices, there exists a fully dynamic algorithm that maintains an implicit $4\cdot{}2^{\alpha}$ colouring with an amortized update of $O(\operatorname{polylog}{n})$ and a query time of $O(\alpha\log{n})$.
\end{corollary}
Finally, we modify an algorithm of Brodal \& Fagerberg \cite{Brodal99dynamicrepresentations} so that it maintains an acyclic out-orientation.
\begin{theorem}
\label{thm:ModificationBroFa}
Given an initially empty and dynamic graph $G$ undergoing an arboricity $\alpha_{max}$ preserving sequence of insertions and deletions, there exists an algorithm maintaining an acyclic $(2\alpha_{max}+1)$ out-degree orientation with an amortized insertion cost of $O(\alpha_{max}^2)$, and an amortized deletion cost of $O(\alpha_{max}^2 \log{n})$.
\end{theorem}

\subparagraph{Paper outline} We recall related work below. In Section~\ref{sec:prelim} we recall preliminaries and and show how to explicitly 2-orient dynamic forests as a warm.up; Section~\ref{sec:Frac} outlines the proof of Theorem~\ref{OutOrientationMain} and Corollary~\ref{cor:ColouringSecondary}; in Section~\ref{sec:DynFo}, we outline how to show Theorem~\ref{ArbDecompMain} and Corollary~\ref{cor:ColouringMain}; and Section~\ref{sec:BroFa} is dedicated to Theorem~\ref{thm:ModificationBroFa}. 
Finally, in Section~\ref{sct:colouring} we give the missing details in the proof of Corollaries \ref{cor:ColouringMain} and \ref{cor:ColouringSecondary}.

\subsection{Related Work}
\subparagraph{Dynamic colouring:}
Barba et al.~\cite{Barba} 
give algorithms for the dynamic recolouring problem, and 
show that $c$-colouring a dynamic forests incurs $\Omega(n^{\frac{1}{c(c-1)}})$ recolourings per update. 
Solomon \& Wein 
give improved trade-offs between update time and recolourings 
and give a deterministic dynamic colouring algorithm 
parametrized by the arboricity $\alpha$, using $O(\alpha^2\log{n})$ colours with $O(1)$ amortized update time~\cite{10.1145/3392724}. 
Henzinger et al.\ \cite{henzinger2020explicit} introduced the study of implicit colouring of sparse graphs in order to circumvent the explicit lower bound of Barba et al.\ \cite{Barba}; 
they maintain an implicit colouring using $2^{O(\alpha)}$ colours, with $O(\log^3{n})$ update time and $O(\alpha\log{n})$ query-time. 
Bhattacharya et al.\ \cite{DBLP:conf/soda/BhattacharyaCHN18} studied the dynamic colouring problem parameterized by the maximum degree $\Delta$, presenting a $(1+o(1))\Delta$ colouring algorithm with $O(\text{polylog}(\Delta))$ update time, and a randomized $\Delta+1$ colouring with expected amortized update time $O(\log{\Delta})$. 
This randomized result was subsequently improved independently by Bhattacharya et al.~\cite{DBLP:journals/corr/abs-1910-02063} and Henzinger \& Peng~\cite{DBLP:conf/stacs/Henzinger020} achieving $O(1)$ amortized update time (respectively, w.h.p. and expected).

\subparagraph{Bounded out-degree orientations:}
Much of the work with respect to bounded out-degree orientations has gone into either $1)$ statically computing bounded out-degree orientations with the minimum (or close to it) out-degree \cite{10.1145/62212.62252,https://doi.org/10.1002/net.3230120206,doi:10.1137/1.9781611974317.10,db16bc9212da408a92161ab5d03b6a88}, or $2)$ dynamically maintaining bounded out-degree orientations with efficient updates \cite{Brodal99dynamicrepresentations,kopelowitz2013orienting,10.1145/3392724,10.1016/j.ipl.2006.12.006,berglinetal:LIPIcs:2017:8263}, but allowing weaker guarantees on the minimum out-degree (see Table \ref{TableB} for an overview).
\begin{table}[]
\begin{tabular}{lllll}
Reference           & Out-degree                & Update time                                       & $\alpha$   \\ \hline
Brodal \& Fagerberg \cite{Brodal99dynamicrepresentations} & $2(\alpha+1)$             & $O(\alpha+\log{n})$ am.                                  & fixed                               \\
Kopelowitz et al.\ \cite{kopelowitz2013orienting}   & $\beta\alpha + \log_{\beta}{n}$     & $O(\beta^2\alpha^{2}+\beta\alpha\log{n})$                             & adaptive                                   \\
He et al.\ \cite{He2014OrientingDG}          & $O(\alpha\sqrt{\log{n}})$ & $O(\sqrt{\log{n}})$ am.                           & fixed                                \\
Berglin \& Brodal \cite{berglinetal:LIPIcs:2017:8263}          & $O(\alpha + \log{n})$ & $O(\log{n})$                            & adaptive                                \\
Henzinger et al.\ \cite{henzinger2020explicit}    & $40\alpha$                & $O(\log^2{n})$ am.                               & adaptive                                       \\ 
Kowalik \cite{10.1016/j.ipl.2006.12.006}.    & $O(\alpha \log{n})$                & $O(1)$ am.                               & fixed                                  \\\hline
New (Thm.\ \ref{OutOrientationMain})           & $(1+\eps)\alpha + 2$              & $O(\log^{5}{(n)}\alpha^2/\eps^6))$                                         & adaptive                                           
\end{tabular}
\caption{\label{TableB} Different dynamic algorithms for maintaining out-orientations.}
\vspace{-4mm}
\end{table}
The current state-of-the-art for exact, static algorithms have running times $O(m^{10/7})$ \cite{Madry} and $O(m\sqrt{n})$ \cite{LeeSidford}. 
Kowalik also gave an algorithm computing a $\lceil (1+\eps) \alpha \rceil$ out-orientation in $\tilde{O}(m\log{n})$ time.

\subparagraph{Arboricity decompositions:}
A lot of work has been put into producing efficient static algorithms for computing arboricity decompositions \cite{10.1145/62212.62252,10.5555/313651.313667,Edmonds1965MinimumPO,https://doi.org/10.1002/net.3230120206} 
(see \cite{blumenstock2019constructive}
for an overview). 
The fastest static algorithm runs in $\tilde{O}(m^{3/2})$ time \cite{10.1145/62212.62252,10.5555/313651.313667}.
Also approximation algorithms have been studied in the static case. 
There exists a linear-time 2-approximation algorithm \cite{DBLP:journals/dam/ArikatiMZ97,10.1016/0020-0190(94)90121-X}. 
Furthermore, Blumenstock \& Fischer provide an algorithm computing a $\lceil (1+\eps) \alpha \rceil+1$ arboricity decomposition in $\tilde{O}(m\log{n})$ time.
Bannerjee et al.\ \cite{Banerjee} provide an $\tilde{O}(m)$ dynamic algorithm maintaining the exact arboricity $\alpha$ of a dynamic graph, and show a lower bound of $\Omega(\log{n})$ for 
dynamically maintaining arboricity.
Henzinger et al.\ \cite{henzinger2020explicit} provide a dynamic algorithm for maintaining a $2\alpha'$ arboricity decomposition, given access to any black box dynamic $\alpha'$ out-degree orientation algorithm
(See Table \ref{TableA}).
\begin{table}[]
\begin{tabular}{lllll}
Reference           & Forests       & Update time                & Uses Lemma\ \ref{lma:ColourForest} & $\alpha$ \\ \hline
Bannerjee et al.\ \cite{Banerjee}    & $\alpha$      &  $\tilde{O}(m)$                          & No                     &   adaptive                     \\
Brodal \& Fagerberg \cite{Brodal99dynamicrepresentations} & $4(\alpha+1)$ & $O(\alpha+\log{n})$  am.             & Yes                    & fixed                    \\
Henzinger et al.\ \cite{henzinger2020explicit}    & $80\alpha$    & $O(\log^{2}{n})$ am.          & Yes                    & adaptive                    \\ \hline
New (Thm. \ref{thm:ModificationBroFa})    & $2(\alpha+1)$ & $O(\alpha^2\log{n})$    am.           & Uses Lemma \ref{lma:split}                     & fixed                    \\
New (Thm.\ \ref{ArbDecompMain})        & $\alpha+2$    & $O(\operatorname{polylog}{n})$      am.                & No                     & adaptive           
\end{tabular}
\caption{\label{TableA} Overview of dynamic algorithms for maintaining arboricity decompositions. Note that applying Lemma \ref{lma:split} to Theorem \ref{thm:ModificationBroFa} gives an arboricity decomposition, since the orientation in Theorem \ref{thm:ModificationBroFa} is acyclic.}
\vspace{-4mm}
\end{table}

\subparagraph{Other related work:} 
Motivated by the problem of finding a densest subgraph, Sawlani \& Wang \cite{DBLP:conf/stoc/SawlaniW20} gave an (implicit) dynamic approximation algorithm for maintaining a $(1+\eps)\rho$ fractional out-degree orientation, where $\rho$ is the maximum subgraph density. 
In order to tune the parameters in the algorithm, they use multiple (partial) copies of the same graph, where each copy has a different estimate of the maximum density of the graph.

Computing near optimal out-orientations and arboricity decompositions has also been studied from a distributed point of view. Barenboim \& Elkin gave a $(2+\eps)$-approximation in \cite{Elkin}. This has since then been improved to $(1+\eps)$-approximations \cite{Harris2,Harrisetal,suetal:LIPIcs:2020:13093,Ghaffari}.
\section{Preliminaries \& Warm-up}
\label{sec:prelim}
Nash-Williams showed that the arboricity $\alpha$ of a graph $G$ satisfies
$\alpha = \lceil \max \limits_{J \subset G} \frac{|E(J)|}{|V(J)|-1} \rceil$ \cite{https://doi.org/10.1112/jlms/s1-39.1.12}. 
A closely related sparsity measure is the \emph{maximum (subgraph) density} $\rho$ defined as
$\rho = \max \limits_{J \subset G}\frac{|E(J)|}{|V(J)|}$. 
Note that $\alpha$ and $\rho$ are numerically very close. 

\subparagraph{Explicit 2-out orientation of dynamic forests}
We begin by considering the simpler problem of orienting the edges of dynamic forests so as to minimise the maximum out-degree of vertices. 
The edges of a forest can be oriented in such a way that the maximum out-degree of a vertex is 1. 
Indeed, we root every tree in the forest arbitrarily and orient all edges towards the roots. 
It is well-known that we can maintain an implicit representation of such a 1-orientation in a dynamic forest using data structures for maintaining information in a dynamic forest as e.g. top trees~\cite{10.1145/1103963.1103966} or Sleator and Tarjans dynamic trees~\cite{10.1145/800076.802464}.
The representation is implicit in the sense that in order to determine the out-edge of a vertex, the dynamic algorithm has to perform some computations.
This is achieved by maintaining the dynamic forests using for instance top trees~\cite{10.1145/1103963.1103966}. Each top tree is then arbitrarily rooted, and one can determine an out-edge, if it exists, of a vertex $v$ by finding the first edge on the unique $v$-to-root-path.
This solution has worst case query and update time of $O(\log{n})$.

A natural next question is if we can maintain such an orientation explicitly i.e.\ in such a way that we at all times explicitly store which way an edge is oriented, whilst still achieving logarithmic update times. 
In turns out, we cannot dynamically maintain such a 1-out orientation in logarithmic time, not even for a dynamic set of paths. 
Indeed, consider two 1-out oriented paths $P_{1}$, $P_{2}$ both of length $k$. 
$P_{i}$ has exactly one vertex of out-degree 0 that all edges are oriented towards, and this vertex is of distance at least $k/2$ from one of the endpoints of $P_{i}$, say $x_{i}$. 
Consequently, adding the edge $x_{1}x_{2}$ forces the reorientation of at least $k/2$ edges. 
Furthermore, this is repeatable: deleting this edge again, determining new vertices $x_{1},x_{2}$ and adding the edge between them, again forces the algorithm to reorient at least $k/2$ edges. 
Hence, by setting $k = n/2$ we find that there exists an $n$-vertex dynamic set of paths together with a sequence of $O(1)$ moves forcing any algorithm maintaining a 1-out orientation of the paths to perform $\Omega(n)$ reorientations for either an insertion or a deletion. 

So in the explicit version of the problem, we have to settle for a 2-out orientation. In this setting, it is straight-forward to solve the problem for dynamic paths: one can orient edges arbitrarily and every vertex will still have out-degree at most two. 
Note that in particular the endpoints of paths have out-degree at most one. 
By dynamically maintaining a decomposition of a dynamic forest into a set of paths together with a set of edges going between the paths such that every such edge is assigned to a unique endpoint of a path, we can extend the above solution to dynamic forests.
Indeed, we orient each path arbitrarily, and orient the inter-path edges away from the unique endpoint the edges were assigned. 
Since these endpoints are endpoints of paths, no vertex receives out-degree more than two. 

In order to maintain this decomposition dynamically, one can maintain the well-known \emph{heavy-light decomposition} of each tree using dynamic trees~\cite{10.1145/800076.802464}. 
For a rooted-tree $T$, we have the notion of parents and children of the vertices. The \emph{parent} of $v$ is the first vertex from $v$ on the $v$-to-root path in $T$. 
The \emph{children} of $v$ are all neighbours of $v$ that are not the parent of $v$. 
A \emph{heavy child} $w$ of $v$ is then a child of $v$ such that the sub-tree of $T$ rooted a $w$ contains more than half of the vertices of the sub-tree rooted at $v$.
The heavy children in $T$ induces \emph{heavy edges} going from a vertex to its heavy child, and \emph{light edges} going from a vertex to a non-heavy child.
Every root-to-leaf path then contains at most $O(\log n)$ light edges. 
The heavy edges form the desired paths, and the light edges can be assigned to the endpoint that is a child of the other endpoint. 

Sleator and Tarjan~\cite{10.1145/800076.802464} showed how to maintain such a heavy-light decomposition in $O(\log n)$ worst case update time. 
The algorithm maintains a decomposition of the edges into \emph{solid} and \emph{dashed} edges and at the end of each update the solid edges coincide exactly with the heavy edges, and the dashed edges with the light edges.
With $O(1)$ overhead we can orient every solid edge arbitrarily and every dashed edge from the child endpoint to the parent endpoint. 
An update changes at most $O(\log n)$ solid edges to become dashed edges, and so the update time is still $O(\log{n})$.

To argue correctness, we need to argue that at the end of each update every dashed edge is oriented from the child towards the parent. 
Whenever a solid edge is turned into a dashed edge, the orientation is chosen such that this invariant is ensured.
Therefore, the only problematic case might be if a dashed edge has the parent-child relationship of its endpoints switched, without the edge having been a solid edge in the meantime. 
However, the only way a dashed edge $e$ has the parent-child relationship of its endpoints switched, is if the root $r_{1}$ of a tree $T$ is switched to be $r_{2}$ such that $e$ is on the unique path between $r_{1}$ and $r_{2}$ in $T$. 
For the root to be switched to $r_2$, every edge on the path from $r_1$ to $r_2$ is turned solid, hence $e$ is turned solid, and so the problematic case never occurs. 
As such, we have shown:
\begin{lemma}
\label{lma:orienttrees}
There exists a fully-dynamic algorithm maintaining an explicit 2-out orientation of an $n$-vertex dynamic forest with $O(\log{n})$ worst-case update time.
\end{lemma}

\subparagraph{Fractional Out-degree Orientations}
We will obtain a low-bounded out-degree orientation in the general case by deterministically rounding, what we shall refer to as a \emph{fractional out-degree orientation}. 
Here, the orientation problem is relaxed so that the edges are allowed to be assigned partially to each end-point, and the goal is to compute an orientation such that the maximum total load assigned to a vertex is minimized. 
A formal definition is as follows:
\begin{definition}
\label{def:FracArb}
A \emph{fractional $\alpha'$-bounded out-degree orientation} $O$ of a graph $G$ is a pair of variables $X_{e}^{u}$,$X_{e}^{v} \in [0,1]$ for each edge $e=uv \in E(G)$ s.t. the following holds:
\begin{center}
    \begin{enumerate}
        \item $\forall e = uv \in E(G)$: $X_{e}^{u}+X_{e}^{v} = 1$
        \item $\forall v \in V(G)$: $s(v) = \sum \limits_{e:v \in e} X_{e}^{v} \leq \alpha'$
    \end{enumerate}
\end{center}
If furthermore $X_{e}^{u}$,$X_{e}^{v} \in \gamma^{-1} \cdot{} \mathbb{Z}$ for all $e \in E(G)$, we say that $O$ is a \emph{$(\gamma,\alpha')$-orientation}.
\end{definition}
In particular, an $\alpha'$-bounded out-degree orientation is just a $(1,\alpha')$-orientation. 
We think of $s(v)$ as the load on vertex $v$, and $\alpha'$ as an upper bound on the allowed vertex load.
The $\gamma$-parameter underlines the fact that we wish to discretise the fractional loads on edges to rational loads. 
If one does so in a symmetric manner for each edge, one can view a $(\gamma,\alpha')$-orientation of a graph $G$ as a $(1,\gamma \alpha)$-orientation of $G^{\gamma}$, where we define $G^{\gamma}$ to be $G$, where every edge is replaced by $\gamma$ copies. 
For an edge $e = uv \in E(G)$, we denote by $B_e$ the bundle of $\gamma$ edges representing $e$ in $G^{\gamma}$. If $G^{\gamma}$ is oriented, we denote by $B_{e}^{u}$ the bundle of edges oriented $u \rightarrow v$. 
Since the copies of $e$ in $B_{e}$ are identical, we only care about the size of $B_{e}^{u}$, and not which copies of $e$ it contains. Hence:
\begin{observation}
For a graph $G$, there is a natural bijection (up to symmetry) between $(1,\gamma\cdot{}\alpha')$ orientations of $G^{\gamma}$ and $(\gamma,\alpha')$-orientations of $G$.
\end{observation}
In light of this observation, we shall use these two descriptions interchangeably, and in some cases we shall refer to the same orientation as being both a $(1,\gamma \alpha')$-orientation of $G^{\gamma}$ and a $(\gamma, \alpha')$-orientation of $G$. 
We follow the approaches of Sawlani \& Wang \cite{DBLP:conf/stoc/SawlaniW20} and Kopelowitz et al.\ \cite{kopelowitz2013orienting}, so we repeat the following:
\begin{definition}[\cite{DBLP:conf/stoc/SawlaniW20}]
Given a $(1,\alpha')$-orientation of a graph $G^{\gamma}$, we say that an edge $u\rightarrow v \in E(G)$ is $\eta$-\emph{valid} if 
$ s(u)-s(v) \leq \eta $ and \emph{$\eta$-invalid} otherwise.
If also $s(v)-s(u) \leq \eta$, we say that $e = uv \in G$ is \emph{doubly $\eta$-valid}.
Furthermore, if
$ s(v)-s(u) \leq -\eta/2 $
we say that $e$ is an $\eta$-tight out-edge of $u$ and an $\eta$-tight in-edge of $v$.
\end{definition}
Note that if $u\rightarrow v$ is $\eta$-invalid, then $s(u)-s(v) > \eta$ and so $-\eta > s(v)-s(u)$, so $uv$ is $\eta$-tight.
Also note that $\eta$ will only differ from $1$ in Appendix \ref{apx:schedule}.
\begin{definition}[\cite{DBLP:conf/stoc/SawlaniW20} Def.\ 3.5]
\label{def:chain}
A \emph{maximal $\eta$-tight chain from $v$} is a path of $\eta$-tight edges $v_0v_1, \dots v_{k-1}v_{k}$, such that $v_0 = v$ and $v_k$ has no $\eta$-tight out-edges. 

A \emph{maximal $\eta$-tight chain to $v$} is a path of $\eta$-tight edges $v_0v_1, \dots v_{k-1}v_{k}$, such that $v_k = v$ and $v_0$ has no $\eta$-tight in-edges. 
\end{definition}
\begin{lemma}[Implicit in \cite{DBLP:conf/stoc/SawlaniW20}]
\label{lma:chainflip}
Inserting an $\eta$-valid edge oriented $u \rightarrow v$ and reorienting a maximal $\eta$-tight chain from $u$ will $\eta$-invalidate no $\eta$-valid edges.

Deleting an edge oriented $u\rightarrow v$ and reorienting a maximal $\eta$-tight chain to $u$ will $\eta$-invalidate no $\eta$-valid edges.
\end{lemma}
\begin{proof}
Inserting the edge increases $s(w)$ by one for some $w$. This will never $\eta$-invalidate any in-edges of $w$, but it might invalidate an out-edge of $w$.
Suppose the edge $w\rightarrow w'$ is invalidated if $s(w)$ increases by 1. 
Then before $s(w)$ increases, we must have $s(w) \geq s(w')+\eta$ and so the edge $ww'$ is tight, contradicting that we reorient a maximal $\eta$-tight chain from $u$. The other statement is similar.
\end{proof}
\begin{remark}
\label{rmk:chainlength}
Note that a maximal $\eta$-tight chain has length at most $\frac{2\cdot{}\max_{v} s(v)}{\eta}$.
Indeed, each time we follow an $\eta$-tight out-edge the load on the vertex increases by at least $\eta/2$.
\end{remark}
If every edge is $\eta$-valid, Sawlani \& Wang say that the orientation is \emph{locally $\eta$-stable}. Kopelowitz et al.\ show the following guarantees for locally $1$-stable orientations,
where we, for ease of notation, define $\Delta^{+} := (1+\eps)\alpha\gamma+\log_{(1+\eps)}{(n)}$:
\begin{lemma}[Implicit in \cite{kopelowitz2013orienting}]
\label{lma:Kop-EtaValid}
If every edge in $G^{\gamma}$ is $1$-valid, then 
$\max_{v} s(v) \leq \Delta^{+}$.
\end{lemma}
\begin{proof}
The proof is completely synchronous to that of Theorem 2.2 in \cite{kopelowitz2013orienting}, but we present it here for completeness. 

Assume every edge is $1$-valid and suppose for contradiction that $d^{+}(v) > \Delta^{+}$ in $G^{\gamma}$. Consider the $i$'th distance class of $v$ in $G^{\gamma}$, $V_{i}$, i.e.\ $V_{i}$ is the set of vertices reachable from $v$ via directed paths of length no more than $i$. 
Then for $1 \leq i \leq \log_{(1+\eps)}{n}$ and for every vertex $w \in V_{i}$, we have: 
\[
d^{+}(w) \geq d^{+}(v)-i \geq \Delta^{+}-i \geq (1+\eps)\alpha\gamma
\]
Now, it follows by induction on $i$ that $|V_{i}| \geq (1+\eps)^{i}$ for all $1 \leq i \leq \log_{(1+\eps)}{n}$. 
Indeed, by the Theorem of Nash-Williams~\cite{https://doi.org/10.1112/jlms/s1-39.1.12} we have $\alpha(V_{1}) \leq \gamma \alpha(G)$, and so $|V_{1}| > \Delta^{+}/(\gamma\alpha) \geq  (1+\eps)$.
For the induction step, note that the number of out-edges of vertices in $V_{i-1}$ is at least $|V_{i-1}|\cdot{}(1+\eps)\alpha\gamma$. Now, the induction step follows by applying the Theorem of Nash-Williams \cite{https://doi.org/10.1112/jlms/s1-39.1.12} to get that $\alpha(V_{i}) \leq \alpha(G)\gamma$ and that $|V_{i}|-1 \geq (1+\eps)|V_{i-1}| > (1+\eps)^{i}$. 
\end{proof}

\subparagraph{Implicit orientations}
We are interested in maintaining a fractional out-degree orientation in which the fractional orientation of edges allow us to 'round' the fractional orientation to a low out-degree orientation. 
We are interested in two properties: first of all the maximum load of a vertex should be low, and second of all many of the edges should have either $X_{e}^{u}$ or $X_{e}^{v}$ close to 1, so that a naive rounding strategy does not increase the load of a vertex by much. 
By Lemma \ref{lma:Kop-EtaValid}, if we ensure that the orientation is locally $1$-stable, then we get an upper bound on the maximum vertex load. 
In order to ensure the second property, we redistribute load along cycles without breaking local stability.
Our algorithm has two phases. 
A phase for inserting/deleting edges in a manner that $\eta$-invalidates no edges, thus ensuring the first property, and a second phase for redistributing load along edges in order to ensure that the orientation also has the second property. 
In order for these two phases to work (somewhat) independently, we think of each phase as having implicit access to the orientation; that is the insertion/deletion algorithm might have to pay a query cost in order to identify the precise fractional load of an edge or neighbourhood of a vertex.
\begin{definition}
\label{def:implicitaccess}
An algorithm on an $n$ vertex dynamic and oriented graph has implicit $(|L|,q(n))$ access to an orientation, if it has access to:
\begin{enumerate}
    \item Operations for querying and changing fractional loads of edges in $O(\log{n})$ time.
    \item A query that returns a list containing a superset of all neighbours of a vertex that have changed status as in- or out-neighbour, since the last time the query was called on this vertex. 
    The list should have length $\leq |L|$ and the query should run in $O(q(n))$ time.
\end{enumerate}
\end{definition}

\subparagraph{Implicitly Accessing Orientations}
\label{sct:SWKop}
In this section, we outline how to modify the algorithm of Kopelowitz et al.\ \cite{kopelowitz2013orienting} to run on $G^{\gamma}$ and to support implicit access to the orientation.
We describe the modifications here and give pseudocode Appendix \ref{apx:KopSW}.
The ideas presented here are not new; they arise in \cite{kopelowitz2013orienting} and \cite{DBLP:conf/stoc/SawlaniW20}, but we present them for completeness. 

We think of the algorithm as being run on $G^{\gamma}$ for some $\gamma$ to be specified later. 
We think of each edge $e \in G$ as $\gamma$ copies in $G^{\gamma}$, but in practice we only store $e$ along with counters $|B_{e}^{u}|$, $|B_{e}^{v}|$ denoting the number of copies oriented in each direction.
Now, we wish to run the algorithm from \cite{kopelowitz2013orienting} in order to insert/delete each copy of an edge one-by-one.
This algorithm inserts/deletes a copy of an edge in $G^{\gamma}$ using Lemma \ref{lma:chainflip} with $\eta = 1$.
We identify a tight chain from $u$ by continuously looking at all out-neighbours and following tight out-edges, until the chain becomes maximal. 
We use max-heaps, stored at each vertex, to identify maximally tight chains to $u$.
Since we only have implicit access to the orientation, we have to first process the list of possible changes to in- and out-neighbours before trying to identify the next tight edge.
Furthermore, when we reorient said chains, we have to access the fractional load of each edge on the chain, before we can change it. Hence, we have:
\begin{theorem}[implicit in \cite{kopelowitz2013orienting}]
\label{thm:Kop-main}
Given implicit $(|L|,q(n))$ access to an orientation with $\max_{v} s(v) \leq \Delta^{+}$, there exists an algorithm that can insert and delete edges from the orientation without creating any new $1$-invalid edges.
The algorithm has worst-case insertion time of $O(\gamma\cdot{}\Delta^{+}(\Delta^{+}+\log(n)(|L|+1)+q(n))))$ and a worst case deletion time of $O(\gamma\cdot{}\Delta^{+}(\log(n)(|L|+1)+q(n))))$.
\end{theorem}
\begin{proof}
See Appendix \ref{apx:KopSW}.
\end{proof}
\begin{remark}
\label{rmk:fracchange}
Each insertion/deletion of a copy of an edge in $G^{\gamma}$ with $\max_{v} s(v) \leq \Delta^{+}$ changes the load of at most $O(\Delta^{+})$ edges. 
Indeed, we only change the load of edges on tight chains (and potentially one new edge), so the statement follows from Remark \ref{rmk:chainlength}.
\end{remark}

\subparagraph{Scheduling Updates}
Some of our algorithms need upper bounds on the arboricity to provide the ensured guarantees.
This is, however, not as limiting a factor as one might think, if we are willing to settle for implicit algorithms.
In this section we describe how to use the algorithm of Sawlani \& Wang \cite{DBLP:conf/stoc/SawlaniW20} to schedule updates to $O(\log{n})$ different copies of a graph such that each copy satisfies different density constraints. 
Here, we describe the main ideas behind the algorithm, and in Appendix \ref{apx:schedule}, we paraphrase the ideas in more details.

Sawlani \& Wang \cite{DBLP:conf/stoc/SawlaniW20} maintain a fractional out-orientation of a graph $G$ by using an algorithm similar to Theorem \ref{thm:Kop-main} to insert and delete edges in $G^{\gamma}$. 
By allowing $\eta$ to scale with the maximum density $\rho$ of $G$, they are able to make the update time independent of the actual value of $\rho$, provided that they have accurate estimates of $\rho$.
By using $\log{n}$ copies of $G$ - each with different estimates $\rho_{est}$ of $\rho$, they are able to at all times keep the copy where $\rho_{est} \leq \max_{v} s(v) < 2\rho_{est}$ fully updated.
They call this copy the \emph{active} copy. 
Similar to Remark \ref{rmk:chainlength} they observe that one can safely insert an edge $uv$ in a copy where at least one of $s(u)$ or $s(v)$ is below $2\rho_{est}$.
If, however, this is not the case, one cannot afford to update the copy.
Sawlani \& Wang resolve this issue by scheduling the updates so that they are only performed, when we can afford to do them.
We can use this algorithm as a scheduler for our algorithms:
We also run our algorithm on $\log{n}$ copies of $G$. 
Whenever the algorithm from Theorem 1.1 in \cite{DBLP:conf/stoc/SawlaniW20} has fully inserted or deleted an edge in a copy, we insert or delete the edge in our corresponding copy. 
Whenever our algorithm is queried, we then use the structure from the currently active copy to answer the query. Hence, we have:
\begin{theorem}[Implicit in \cite{DBLP:conf/stoc/SawlaniW20} as Theorem 1.1]
\label{thm:schedule}
There exists a fully dynamic algorithm for scheduling updates that at all times maintains a pointer to a fully-updated copy with estimate $\rho_{est}$ where $(1-\eps)\rho_{est}/2\leq \alpha(G)<4\rho_{est}$. 
Furthermore, the updates are scheduled such that a copy $G'$ with estimate $\rho'$ satisfies $\alpha(G')\leq 4\rho'$.
The algorithm has amortised $O(\log^{4}(n)/\eps^{6})$ update times.
\end{theorem}
\begin{proof}
See Appendix \ref{apx:schedule}.
\end{proof}
\section{Dynamic Maintenance of Refinements}
\label{sec:Frac}
In this section, we work towards the second goal: an orientation where many edges assign most of their load to one endpoint. To realise this, we introduce refinements of fractional orientations and show how to dynamically maintain them. 
The basic idea is to maintain an orientation such that the fractional load of each vertex is small and such that the edges, that distribute their loads somewhat equally between both endpoints, form a forest. 
This property is nice, if we want to transform our orientation to a bounded out-degree orientation, since all edges outside of the forest almost already have decided on an orientation, and we can 2-orient trees using Lemma \ref{lma:orienttrees}. 
Definition \ref{def:refinement} formalises this idea:
\begin{definition}
\label{def:refinement}
Let $O$ be a $(\gamma,\alpha')$-orientation of a graph $G$. Then $H$ is a $(\delta,\mu)$-refinement of $G$ wrt. $O$ if:
\begin{enumerate}
    \item $V(H) = V(G)$
    \item For all $e = uv \in E(G):$ $X_{e}^{u},X_{e}^{v} \in (\delta,1-\delta)$ implies that $e \in H$.
    \item If $e \in H$, then $X_{e}^{u},X_{e}^{v} \in [\delta-\mu,1-\delta+\mu]$
\end{enumerate}
\end{definition}
The basic idea behind our algorithm is to maintain a refinement that is a forest. Whenever a cycle $C$ occurs in the refinement, we can redistribute the fractional loads along the cycle so as to not change $s(v)$ for any $v \in C$, but such that an edge of $C$ does not satisfy condition 2.\ in Definition \ref{def:refinement}. 
\begin{figure}[]%
    \centering
    \subfloat{{\includegraphics[width=8cm]{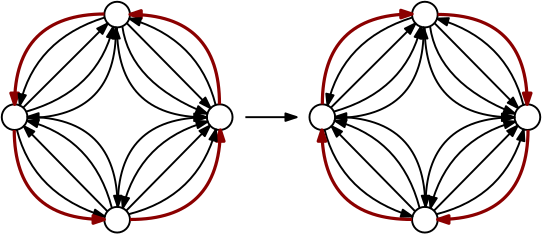} }}%
    \caption{Inverting a cycle. We can reorient copies of edges forming a directed cycle without changing the load of any vertex.}
    \label{fig:augment}%
\end{figure}
Thus we can remove this edge and a again obtain an acyclic refinement with respect to this new orientation. 
Hence, we have the following observation:
\begin{observation}
\label{obs:refinement}
Suppose $1>\delta>\gamma^{-1}+\mu \geq 2\gamma^{-1}>0$.
Let $H$ be $(\delta,\mu)$-refinement of a graph $G$ wrt. some $(\gamma,\alpha')$-orientation $O$ of $G$. Then there exists a $(\delta,\mu)$-refinement of $G$, say $H'$, wrt.\ some $(\gamma,\alpha')$-orientation $O'$ of $G$, such that $H'$ is a forest.
\end{observation}
\begin{proof}
Suppose $H$ is not one such $(\delta,\mu)$-refinement i.e.\ $H$ is not a forest. Let $C = v_0, v_1, \dots, v_c, v_0$ be a cycle in $H$. 
Set 
\begin{equation*}
    l(C) = \min \limits_{i} \{ \min \{X_{v_{i}v_{i+1}}^{v_i},1-X_{v_{i}v_{i+1}}^{v_i}\} \}
\end{equation*}
Now if $l(C) \leq \delta$, we can just remove the edge minimizing $l(C)$ from $H$, so suppose this is not the case.
If $X^{v_i}_{v_{i}v_{i+1}} = l(C)$ for some $i$, we set $X^{v_j}_{v_{j}v_{j+1}} = X^{v_j}_{v_{j}v_{j+1}}-(l(C)-\delta+\mu)$ and $X^{v_{j+1}}_{v_{j}v_{j+1}} = X^{v_{j+1}}_{v_{j}v_{j+1}}+(l(C)-\delta+\mu)$ for all $j$. 
Otherwise, $1-X^{v_j}_{v_{i}v_{i+1}} = l(C)$ for some $i$, and we set $X^{v_j}_{v_{j}v_{j+1}} = X^{v_j}_{v_{j}v_{j+1}}+(l(C)-\delta+\mu)$ and $X^{v_{j+1}}_{v_{j}v_{j+1}} = X^{v_{j+1}}_{v_{j}v_{j+1}}-(l(C)-\delta+\mu)$  for all $j$.
Observe that this change of the fractional orientations on the cycles does not violate the 2 conditions in Definition \ref{def:FracArb}, so this yields a new $(\gamma,\alpha')$-orientation $\hat{O}$ of $G$. 
Furthermore, $H$ is still a $(\delta,\mu)$-refinement of $G$. 
Indeed, we only change the fractional load of edges in $C$, and for these edges, after the changes, we have: 
\begin{align*}
    X^{v_j}_{v_{j}v_{j+1}} &\geq \delta-\mu \\
    X^{v_j}_{v_{j}v_{j+1}} &\leq  1-\delta+\mu
\end{align*}
Furthermore, for the $i$ minimizing $l(C)$ one of the inequalities holds with equality, allowing us to remove it from $H$ and obtain a new, smaller $(\delta,\mu)$-refinement of $G$ wrt.\ the new orientation. 
Continuing like this eventually turns $H$ in to a forest. 
\end{proof}
Note that if every edge of a graph $G$ is $\eta$-valid, then an edge $e = uv \in G$ can only distribute its load somewhat evenly between $u$ and $v$, if $s(u)$ and $s(v)$ are approximately the same. 
This implies that $e$ is actually doubly $\eta$-valid:
\begin{observation}
\label{obs:validity}
If every edge in $G^{\gamma}$ is $\eta$-valid, then every edge of a $(\delta,\mu)$-refinement $H$ of $G$ with $1>\delta > \gamma^{-1}+\mu \geq 2\gamma^{-1}$ is doubly $\eta$-valid.
\end{observation}
\begin{proof}
If $e = uv \in H$, then $X_{e}^{u},X_{e}^{v} \in [\delta-\mu,1-\delta+\mu]$ so $X_{e}^{u},X_{e}^{v} > 0$. 
Hence, in $G^{\gamma}$ we have $|B_{e}^{u}|, |B_{e}^{v}| > 0$, and so there is at least one edge oriented $u \rightarrow v$ and at least one edge oriented $v \rightarrow u$ in $G^{\gamma}$. 
Since all edges are $\eta$-valid, the statement follows. 
\end{proof}
Since the redistribution of fractional load of edges along a cycle does not change the load $s(v)$ of any vertex $v$, performing the redistribution from Observation \ref{obs:refinement} $\eta$-invalidates no edges.
\subsection{The algorithm}
As outlined earlier, our algorithm has 2 phases. In the first phase, we will insert and delete edges without $\eta$-invalidating any edges. 
We do this using the algorithm from Section \ref{sct:SWKop}.
In the second phase, we examine all of the edges, whose fractional load was altered in phase 1. 
These edges might need to enter or exit $H$, depending on their new load. 
If such an insertion in $H$ creates a unique cycle, we remove it as described in Observation \ref{obs:refinement}. 
More precisely, the algorithm works as follows:
\begin{enumerate}
    \item Insert (delete) $\gamma$ copies of $e$ into $G^{\gamma}$ one at a time using a phase I algorithm from Section \ref{sct:SWKop}. 
    Whenever a copy of an edge $f \in E(G)$ is reoriented in phase I, we push $f$ onto a stack $Q$. If $e$ is deleted in $G$, we also remove $e$ from $H$.
    \item When all $\gamma$ copies of $e$ are inserted, we set $g = uv = \texttt{pop}(Q)$ and update $H$ as follows until $Q$ is empty:
    \begin{itemize}
        \item If $g \in H$ and $X^{u}_g \in (\delta,1-\delta)$, we update the weight of $g$ in $H$ to match that of $G^{\gamma}$.
        \item If $g \in H$ and $X^{u}_g \notin (\delta,1-\delta)$, we remove $g$ from $H$.
        \item If $g \notin H$ and $X^{u}_g \in (\delta,1-\delta)$, we push $g$ onto a new stack $S$.
        \item If $g \notin H$ and $X^{u}_g \notin (\delta,1-\delta)$, we do nothing.
    \end{itemize}
    \item After processing all of $Q$, $H$ together with the edges in $S$ form a $(\delta,\mu)$-refinement of $G$. We now process each edge $h = uv \in S$ as follows:
    \begin{itemize}
        \item If $u,v$ are not in the same tree in $H$, we insert $h$ into $H$.
        \item Otherwise, $u,v$ are in a unique cycle $C$ in $H$. We update the weights along $C$, locate an edge $wz$ along $C$ with $X_{wz}^{w},X_{wz}^{z} \notin (\delta,1-\delta)$  and remove it from $H$. If $uv \neq wz$, we insert $uv$ into $H$.
        \item Finally, we update $B_{wz}^{w},B_{wz}^{z}$ in $G-H$ to match the weights $wz$ had in $H$.
    \end{itemize}
\end{enumerate}
Since only edges from $Q$ can enter $S$, we have the following Observation:
\begin{observation}
\label{obs:stacks}
Let $S_{\max}$ and $Q_{\max}$ denote the maximum size of the stacks above during an insertion or a deletion. Then we have $S_{\max} \leq Q_{\max} \leq T$, where $T$ is the total number of edges whose fractional orientation are altered during an insertion or a deletion.
\end{observation}
Furthermore, Observations \ref{obs:refinement} and \ref{obs:validity} and Theorem \ref{thm:Kop-main} imply the invariants:
\begin{invariant}
\label{inv:valid}
Under the orientation induced by $H$ for edges in $E(H)$ and by $G-H$ for edges in $E(G-H)$, every edge in $E(G)$ is $\eta$-valid.
\end{invariant}
\begin{invariant}
\label{inv:forest}
$H$ is both a $(\delta,\mu)$-refinement and a forest.
\end{invariant}

\subsection{Implementing updates}
Since we maintain the invariant that $H$ is a forest, we can use data structures for maintaining information in fully dynamic forests to store and update $H$:
\begin{lemma}[Implicit in \cite{10.1145/1103963.1103966}] 
\label{lma:toptrees}
Let $F$ be a dynamic forest in which every edge $e = wz$ is assigned a pair of variables $X_{e}^{w}, X_{e}^{z} \in [0,1]$ s.t. $X_{e}^{w} + X_{e}^{z} = 1$. Then there exists a data structure supporting the following operations, all in $O(\log{|F|})$-time:
\begin{itemize}
    \item $link(u,v,X_{uv}^{u}, X_{uv}^{v})$: Add the edge $uv$ to $F$ and set $X_{uv}^{u}, X_{uv}^{v} = 1-X_{uv}^{u}$ as indicated.
    \item $cut(u,v)$: Remove the edge $uv$ from $F$.
    \item $connected(u,v)$: Return $\operatorname{true}$ if $u,v$ are in the same tree, and $\operatorname{false}$ otherwise.
    \item $add\_weight(u,v,x)$: For all edges $wz$ on the path $u\dots wz \dots v$ between $u$ and $v$ in $F$, set $X_{wz}^{w} = X_{wz}^{w}+x$ and $X_{wz}^{z} = X_{wz}^{z}-x$. 
    \item $min\_weight(u,v)$: Return the minimum $X_{wz}^{w}$ s.t.\ $wz$ is on the path $u\dots wz \dots v$ in $F$.
    \item $max\_weight(u,v))$: Return the maximum $X_{wz}^{w}$ s.t.\ $wz$ is on the path $u\dots wz \dots v$ in $F$.
\end{itemize}
\end{lemma}
\begin{proof}
See Appendix \ref{apx:Refinements}.
\end{proof}
Note that using non-local search as described in \cite{10.1145/1103963.1103966}, one can also locate the edges of minimum/maximum weight in $O(\log{|F|})$-time. The Lemma also shows that we can process an edge in $Q$ in $O(\log{n})$-time.
\begin{observation}
\label{obs:accEdge}
We can access and change the fractional load of $e \in H$ in time $O(\log{n})$. We can do the same for $e \in G-H$ in $O(1)$ time, since these loads are not stored in top trees.
\end{observation}
To process an edge $uv \in S$ creating a cycle $C$ in the $(\delta,\mu)$ refinement $H$, we do as follows.
Depending on the argument minimizing $l(C) = \mu + \min \{min\_weight(u,v)-\delta,1-\delta-max\_weight(u,v), X_{uv}^{u}-\delta,1-\delta-X^{v}_{uv}\}$, we either add or subtract $l(C)$ to every edge in $C$. We determine and remove the edge that minimized $l(C)$ from $H$. Thus we can process an edge in $S$ in $O(\log{n})$-time.
Finally, if $\mu > \gamma^{-1}$ then every edge in $S \cup H$ has at least one copy in $G^{\gamma}$ pointing in each direction both before and after the inversion of a cycle. 
Hence, no vertex receives any new in- nor out-neighbours. Since inverting a cycle does not change the load of any vertex, we need not update any priority queues for the insertion/deletion algorithm. Hence, we do not have to return any lists and so $|L| = q(n) = 0$.

\subsection{Conclusions}
\begin{theorem}
\label{thm:RefinementKop}
Suppose $1>\delta>\gamma^{-1}+\mu> 2\gamma^{-1}>0$, $\eps > 0$. 
Then, there exists a dynamic algorithm that maintains a $(\gamma, (1+\eps)\alpha+\gamma^{-1}\log_{(1+\eps)}{n})$-orientation of a dynamic graph $G$ with arboricity $\alpha$ as well as a $(\delta,\mu)$-refinement $H$ of $G$ wrt.\ this orientation such that $H$ is a forest. 
The fractional orientation of an edge can be computed in time $O(\log{n})$, insertion takes worst-case $O(\gamma\cdot{}(\Delta^{+})^2)$ time and deletion takes worst-case $O(\gamma \cdot{}\Delta^{+}\cdot{}\log{}(n))$ time.
\end{theorem}
\begin{proof}
Apply Theorem \ref{thm:Kop-main} for insertion/deletion. 
Note that $|L| = 0$ and $q(n) = 0$.
The time spent repairing $H$ after each insertion/deletion is in $O(\gamma \Delta^{+} \log{n})$ by Remark \ref{rmk:fracchange} and Observations \ref{obs:stacks} and \ref{obs:accEdge}, since we can process an edge from both $Q$ and $S$ in $O(\log{n})$ time. 
Finally, Observation \ref{obs:validity} and the Invariants \ref{inv:valid} and \ref{inv:forest} show correctness of the algorithm.
\end{proof}
Now tuning the parameters of Theorem \ref{thm:RefinementKop}, rounding edges in $G-H$ and 2-orienting $H$ yields Theorem \ref{OutOrientationMain}. We restate Theorem \ref{OutOrientationMain} for convenience:
\begin{theoremS}[Theorem \ref{OutOrientationMain}]
For $1> \eps > 0$, there exists a fully-dynamic algorithm maintaining an explicit $((1+\eps)\alpha+2)$-bounded out-degree orientation with worst-case insertion time $O(\log^{3}{n}\cdot{}\alpha^2/\eps^6)$ and worst-case deletion time $O(\log^{3}{n}\cdot{}\alpha/\eps^4)$.
\end{theoremS}
\begin{proof}
Using Theorem \ref{thm:RefinementKop} with $\eps' = \eps/20$, $\gamma = \log{n}/\eps'^2$, $\delta = 2\gamma^{-1}$ and $\mu = \gamma^{-1}$, we can with only $O(1)$ overhead deterministically round any edge in $G-H$ to point away from the vertex to which it assigns the highest load.
This gives an out-degree upper bounded by 
\begin{align*}
    &((1+\eps')\alpha+ \gamma^{-1}\log_{1+\eps'}(n)) /(1-2\gamma^{-1}) \\
& \leq ((1+\eps')\alpha+ \gamma^{-1}\eps'^{-1}\log(n)) \cdot{}(1+4\gamma^{-1}) \\
&\leq \alpha(1+\eps'+\gamma^{-1}\eps'^{-1}\log(n)+4\gamma^{-1}+4\gamma^{-1}\eps'+4\gamma^{-2}\eps'^{-1}\log(n))\\
& \leq \alpha(1+2\eps'+4\eps'^2+4\eps'^{3}+4\eps'^{5}) \leq \alpha(1+\eps)
\end{align*}
where we used the fact that for $r \leq 1/2$, we have $1/(1-r) \leq 1+2r$.
Note that this is an upper bound on the out-degree, so the actual out-degree is at most the floor of this expression.
Finally, we also run the algorithm from Lemma \ref{lma:orienttrees} to 2-orient $H$. This gives at most 2 extra out-edges per vertex, and takes time $O(\log(n)\log(n)/\eps^{-2}\cdot{}\alpha \log{n}) = O(\log^{3}(n)\alpha \eps^{-2})$, since we have at most $O(\gamma\Delta^{+})$ insertions and deletions into $H$ per insertion into $G$, and we can insert each such edge into the rounding scheme on $H$ spending $O(\log(n))$-time. 
\end{proof}
By naively rounding $G-H$ in Theorem \ref{thm:RefinementKop} (for specific values of parameters), and splitting the out-orientation using Lemma \ref{lma:split}, we get an algorithm for dynamically maintaining a decomposition into $\lfloor (1+\eps)\alpha'\rfloor$ pseudoforests and a single forest. 
Applying the colouring techniques described in Section \ref{sct:colouring}, then yields Corollary \ref{cor:ColouringSecondary}. 

\section{Forests}
\label{sec:DynFo}
We begin this section by outlining the main ideas for turning a dynamic low out-orientation into a dynamic low arboricity decomposition. 
Given a dynamic $\alpha'$-bounded out-degree orientation, one can, with very little overhead, split it into $\alpha'$ 1-bounded out-degree orientations using a (slight modification) of an algorithm by Henzinger et al.\ \cite{henzinger2020explicit}. 
Now, given this dynamic pseudoforest partition, we wish to apply the ideas of Blumenstock \& Fischer \cite{blumenstock2019constructive} in order to turn the $\alpha'$ pseudoforests into $\alpha'+1$ forests. 
The main technical challenge of making this process dynamic is the following: the algorithm from \cite{henzinger2020explicit} relies heavily on each vertex having out-degree no more than $1$ in each pseudoforest. 
However, the approach of Blumenstock \& Fischer \cite{blumenstock2019constructive} is to move edges between pseudoforests, showing no regards as to why an edge was placed in a pseudoforest to begin with. 
Hence, if one naively applies this approach on top of the pseudoforest partition, one could potentially ruin the invariant that every vertex has out-degree no more than $1$ in each pseudoforest, causing the algorithm of Henzinger et al.\ \cite{henzinger2020explicit} to fail. 
We tackle this problem in steps. 
First, we show that if we were somehow able to invert the orientations of cycles, then we can make the moves of Blumenstock \& Fischer's approach \emph{faithful} to the degree condition of the pseudoforest algorithm of Henzinger et al.\ \cite{henzinger2020explicit}. 
If we invert orientations along cycles in the pseudoforests, the out-degree of no vertex in the pseudoforests is changed. 
However, if we wish to perform these operations, we will have to do it in a manner that still allows us to maintain the underlying $\alpha'$-bounded out-degree orientation. 
If the cycles are doubly $\eta$-valid, we invert the cycles using Lemma \ref{lma:toptrees}. We do as in Section \ref{sec:Frac}, but this time we add or subtract $1-\delta$ along the cycles. This ensures that every edge on the cycle now prefers the other endpoint, and so is naively rounded to the opposite direction without ending in $H$.
The problem is that we have no guarantee that all edges are doubly $\eta$-valid. 
If an edge is only singly $\eta$-valid, then redistributing the load along a cycle containing this edge causes the edge to become invalid.
However, by Lemma \ref{lma:chainflip}, we can delete such invalid edges and reinsert them again to restore the invariant that all edges are $\eta$-valid.

\subsection{Ideas of Henzinger et al.\ and Blumenstock \& Fischer}
Given an $\alpha'$-bounded out-degree orientation, one can split it into $\alpha'$ pseudoforests by partitioning the edges such that each vertex has out-degree at most one in each partition. 
Then every connected component $P_C$ in a partition is a pseudoforest.
Indeed, $|E(P_C)| \leq |P_{C}|$ since every vertex has out-degree at most one. Hence, there can be at most one cycle in $P_C$. 
This idea is implicit in an algorithm of Henzinger et al.\ \cite{henzinger2020explicit}. Note that we can store each pseudotree as a top tree with one extra edge with only $O(\log{n})$ overhead per operation.
\begin{lemma}[Implicit in \cite{henzinger2020explicit}]
\label{lma:split}
Given black box access to an algorithm maintaining an $\alpha'$-bounded out-degree with update time $T(n)$, there exist an algorithm maintaining an $\alpha'$ pseudo-forest decomposition with update time $O(T(n))$.
\end{lemma}
\begin{proof}
See Appendix \ref{apx:Forests}.
\end{proof}
Using the ideas of Blumenstock \& Fischer \cite{blumenstock2019constructive}, we can represent a pseudoforest $P$ by a pair $(F,M)$ s.t. $F$ is a forest and $M$ is matching, by adding exactly one edge from each cycle in $P$ to $M$. 
Similarly, we can represent a partition of $E(G)$ into pseudoforest $(P_{1}, \dots, P_{k})$ by a pair $(F,M)$ s.t.\ $F = \cup F_i$ and $M = \cup M_{i}$ and $(F_i,M_i)$ represents $P_i$ for all $i$. 

In order to ensure the guarantees of Lemma \ref{lma:split}, we need to maintain the invariant that every vertex has out-degree at most one in every pseudoforest. 
If this is the case, we say that the partition is \emph{faithful} to the underlying orientation.
Blumenstock \& Fischer \cite{blumenstock2019constructive} perform operations on $G[M]$ in order to turn it into a forest. They call $G[M]$ the \emph{surplus graph}. 
Some of the operations, they perform, are described in the following Lemma:
\begin{lemma}[Implicit in \cite{blumenstock2019constructive}]
\label{lma:move1}
Let $(F,M)$ be a faithful representation of a pseudoforest partition of a simple graph $G$ equipped with an $\alpha'$-bounded out-degree orientation. 
If $uv \in M_{i}$ and $vw \in M_{j}$ with $i \neq j$, then there exists an $\alpha'$-bounded out-degree orientation with respect to which the partition gained by swapping $P_{i} \leftarrow P_{i}\cup\{vw\}-\{uv\}$ and $P_{j}\leftarrow P_{j}\cup\{uv\}-\{vw\}$ yields a faithful partition, and $uv \in M_{j}$ resp.\ $vw \in M_{i}$ iff. $uv$ resp.\ $vw$ are on the uni-cycle in their new pseudoforests.

Furthermore, if $wx \in M_{i}$ for some $x$, then $vw$ is not on a uni-cycle in $P_{i}$. 
\end{lemma}
\begin{proof}
See \cite{blumenstock2019constructive} Lemma 2. To see that we can modify the orientation to accommodate the swaps, note that we can always reverse the direction of at most two cycles, without changing the out-degree of any vertex, such that both of the edges swapped are out-edges of $v$. Now swapping the two out-edges ensures that the partition stays faithful to the orientation.
\end{proof}
\begin{figure}[]%
    \centering
    \subfloat{{\includegraphics[width=12cm]{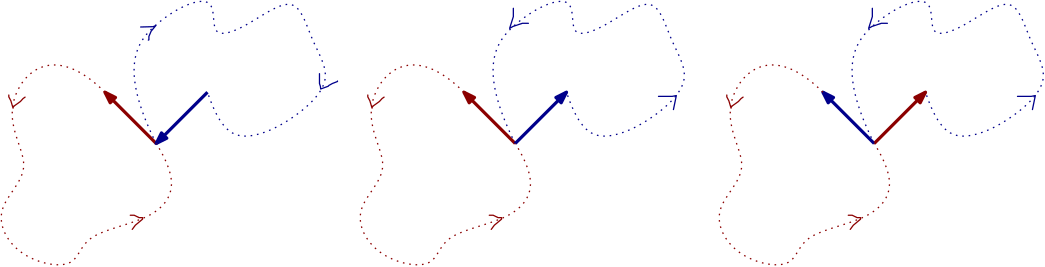} }}%
    \caption{Moving edges between pseudoforests. Before performing the swap, we reorient a cycle so that the swapped edges both are out-edges of their common endpoint}
    \label{fig:switch}%
\end{figure}
Following Blumenstock \& Fischer we note that if $e_1, \dots, e_k$ is a path in a surplus graph $G[M]$ such that $e_1$ and $e_k$ belong to the same matching $M_i$, then we can use the moves from Lemma \ref{lma:move1} to restore colourfulness (see \cite{blumenstock2019constructive} Lemma 3).
The key is that we can move the other edge in $M_i$ towards $e_1$, and then after $O(k)$ switches, we are sure to end up in the furthermore part of Lemma \ref{lma:move1}.
If a surplus graph contains no such paths, Blumenstock \& Fischer say it is a \emph{colourful} surplus graph. They show the following Lemma:
\begin{lemma}[\cite{blumenstock2019constructive}]
\label{lma:move2}
Suppose $J$ is a colourful component of the surplus graph $G[M]$ of a graph $G$. Then for all $v \in J$ there exists an index $i$ s.t. $N_{F_i}(v) \cap J = \emptyset$ and $J \cap M_{i} \neq \emptyset$. 
\end{lemma}
These Lemmas motivate the following approach: use Lemma \ref{lma:move1} to ensure that the surplus graph is always colourful. 
Next use Lemma \ref{lma:move1} to remove any cycles from the surplus graph. 

\subsection{Our algorithm for maintaining dynamic arboricity decompositions}
\label{sct:algo}
Assume that we have an upper bound $\alpha_{max}$ on the arboricity throughout the entire update sequence. The algorithm works roughly as follows:
\begin{enumerate}
    \item Run the algorithm from Theorem \ref{thm:RefinementKop}.
    \item Naively round the orientation of each edge in $G-H$. 
    \item Split the rounded out-degree orientation on $G-H$ into pseudoforests.
    \item Whenever an edge enters or moves between pseudoforests, we push it to a queue $R$.
\end{enumerate}
We process each edge in $e \in R$ as follows: Put $e$ into a pseudoforest. If $e$ completes a cycle in a pseudoforest add it to $G[M]$.
When $e$ enters $G[M]$, we determine if it sits in a colourful component. If it doesn't, we apply Lemma \ref{lma:move1} until all components in $G[M]$ are colourful.
If a non-doubly valid cycle is reoriented in this process, we remove the singly-valid edges from the pseudoforests and add them to $R$. This ensures that the two edges from the same matching that we were trying to separate into two different components, are indeed separated. 
We will later bound the total number of edges pushed to $R$.
If, on the other hand, the component is colourful, $e$ may sit in a cyclic component. 
Then we apply Lemma \ref{lma:move2} to remove the unique cycle. This may create a new non-colourful component, which we handle as before.

In the following, we describe the necessary data structures and sub-routines needed to perform these operations. 
Note that in order to show Lemma \ref{lma:surplusinsertion} in Section \ref{sct:operations}, we use ideas from Sections \ref{sct:nondoubly} and \ref{sct:locating}, so the proof of Lemma \ref{lma:surplusinsertion} is deferred to Section \ref{sct:locating}.

\subsection{Operations on the surplus graph}
\label{sct:operations}
In this section, we assume all cycles are doubly $\eta$-valid. 
In section \ref{sct:nondoubly}, we handle cases where this is not the case.
Assuming that $G[M]$ is both colourful and acyclic, we can insert an edge in $G[M]$ and restore these invariants by performing switches according to Lemmas \ref{lma:move1} and \ref{lma:move2}.
Indeed, after inserting an edge, we can run, for example, a DFS on the component in $G[M]$ to determine if it is colourful. 
If it is not, we locate a path $e_1, \dots , e_k$ such that $e_{1},e_{k} \in M_{i}$.
Then we apply Lemma \ref{lma:move1} to $e_{i-1}$ and $e_{i}$ beginning with $i = k$, until an edge from $M_{i}$ is removed from $G[M]$. Note that this is certain to happen when $e_{1}$ and $e_{3}$ belong to the same pseudoforest.
We continue locating and handling paths until the component becomes colourful. 
If the component is colourful, but not acyclic, we choose a vertex $v$ on the cycle and apply Lemma \ref{lma:move2} to determine a pseudoforest represented in the component in which $v$ is connected to no other vertex in the cyclic component. 
Then we determine a path in the surplus graph between an edge in said pseudoforest and $v$. Now we move the edge in this pseudoforest to $v$ using Lemma \ref{lma:move1}. 
If the edge is removed from $G[M]$ or the path is disconnected, we repeat the process. 
When such an edge is incident to $v$, we switch it with an edge on the cycle.
Finally, we replace it in $M$ with the unique other edge incident to $v$ in the pseudotree that put it in $M$.
Now $G[M]$ is acyclic, but it may not be colourful. 
If this it the case, we repeat the arguments above until it becomes colourful. Note that these moves never create a cycle.
\begin{lemma}
\label{lma:surplusinsertion}
After inserting an edge into $G[M]$, we can restore acyclicity and colourfulness in $O(\alpha^{3}\log^{2}(n))$ time.
\end{lemma}

\subsection{Recovering neighbours}
\label{sct:recovery}
For each vertex, we will lazily maintain which of its out-edges belong to which pseudoforest. 
This costs only $O(1)$ overhead, when actually moving said edges. 
However, whenever we invert a cycle, these edges may change. 
Since the cycles can be long, we can only afford to update this information lazily, whenever the insertion/deletion algorithm determines the new out-neighbours of a vertex.
When this happens, we say the vertex is \emph{accessed}.
Whenever an edge has its fractional load changed via a cycle inversion, it is always changed by the same amount. Hence, we make the following Observation:
\begin{observation}
\label{obs:newNbrs}
Between two accesses of a vertex $v$, the only possible new in-neighbours are the edges which were out-neighbours at the last access of $v$, and the only new out-neighbours are the vertices that are out-neighbours at the current access of $v$.
\end{observation}
\begin{proof}
Consider an edge oriented $u\rightarrow v$ at two consecutive accesses of $v$. 
Then, the fractional orientation of $uv$ is the same as the last time $v$ was accessed. 
Indeed, the fractional orientation can only change, if either the edge is on a maximally tight chain being reoriented, in which case $v$ was accessed, or if $uv$ is on a cycle being inverted. 
Every time, this happens the orientation of $uv$ changes. 
Hence, $uv$ has been inverted exactly the same amount of times in both directions. Since each inversion changes the fractional orientation by the same amount, it follows that the fractional orientation of $uv$ has not changed. 

The same argument also shows that the edges, whose fractional orientation have changed, are exactly the edges that have had there orientation changed since the last time $v$ was accessed. 
These edges are precisely the old and the new out-edges.
\end{proof}
Thus, we can recover exactly which incident edges might have changed in- and/or out-neighbour status from $v$, since the last time $v$ was accessed by the insertion/deletion algorithm. 
To do so, we maintain that each top tree is rooted in the unique vertex, which has out-degree 0, when the underlying orientation is restricted to the tree.  
This ensures that we can recover v's unique out-edge in a pseudoforest by finding the first edge on the unique $v$-to-root path in the top tree.
We maintain this information as follows:
\begin{itemize}
        \item When we $link(u,v)$ with an edge oriented $u\rightarrow v$, we set the root of the new tree to be that of the tree containing $v$.
        \item When we $cut(u,v)$ with an edge oriented $u\rightarrow v$, we set the root of the tree containing $u$ to be $u$ and that containing $v$ to be the same as the old tree.
        \item When we invert the orientation along a cycle originally oriented $u \rightarrow v\rightarrow \dots \rightarrow u$, we change the root from $v$ to $u$.
        \item When we perform a Lemma \ref{lma:move2} swap, we also update the root accordingly. 
\end{itemize}
Note that each update is accompanied by an operation costing $O(\log{n})$ time, so the overhead for maintaining this information is only $O(1)$.
With this information, we can recover the old out-neighbours as the stored out-neighbours, and the new out-neighbours by taking the first edge on the path from $v$ to the root. Hence, we have shown:
\begin{lemma}
\label{lma:implicitAccess}
We can supply each vertex with a query returning a list $L$ of neighbours which might have their status changed in time $O(\alpha \log{n})$. Furthermore, $|L| = O(\alpha)$.
\end{lemma}

\subsection{Non doubly $\eta$-valid cycles}
\label{sct:nondoubly}
If a cycle is not doubly $\eta$-invalid, we still switch the orientation as before, but now we have to fix invalid edges.
Assuming we know which edges have become $\eta$-invalidated, we fix them as follows:
For every invalid edge, we first remove the edge from the pseudoforest it resides in. This has two consequences. Firstly, the algorithm from Lemma \ref{lma:split} might move $O(1)$ edges between pseudoforests, and secondly, we also have to move an edge from the surplus graph back down as a normal edge in the pseudoforest it comes from.
All of the (re)moved edges are pushed to the queue $R$. 
Then, we delete all invalid copies of edges in $G^{\gamma}$, and reinsert them. 
Now, all edges are valid again, and so we continue processing edges in $R$ as described in Section \ref{sct:algo}. If an edge now belongs to $H$, we do not insert it into any pseudoforest.

It is important to note that the second consequence i.e.\ that we remove an edge from $G[M]$, either makes a Lemma \ref{lma:move1} switch successful by removing one of the edges from $M_{i}$, or it removes an edge on one of the at most two paths between edges in $M_{i}$. In this case, we try to locate a second path, and handle it as before. This happens at most once: the component has at most one cycle, and hence at most two paths between two vertices. 
We ascribe the cost of deleting and reinserting invalid edges to the potential in Lemma \ref{lma:potential} that bounds the total number of copies of edges that are inserted into $G^{\gamma}$. 
This cost is not ascribed to the algorithm maintaining $G[M]$. Set $\Delta_{max}^{+} = (1+\eps)\alpha_{max}\gamma + \log_{(1+\eps)}{n}$, we have:
\begin{lemma}
\label{lma:potential}
The total amount of insertions and deletions performed by the insertion/deletion algorithm over the entire update sequence can be upper bounded by 
\begin{equation*}
    O(\frac{\gamma\Delta_{max}^{+}}{\eta}(\Delta_{max}^{+}\cdot{}i+d))
\end{equation*}
\end{lemma}
\begin{proof}
We use the following potential to bound the number of insertions/deletions performed:
\begin{equation*}
    \Phi(O_{i}) = \sum \limits_{v: s(v) \geq 1} (\Delta_{max}^{+}-s(v))^2
\end{equation*}
Note that by Lemma \ref{lma:Kop-EtaValid} $\max s(v) \leq \Delta^{+} \leq \Delta_{max}^{+}$, so an insertion or a deletion of a single copy of an edge in $G^{\gamma}$ increases the potential by $O((\Delta_{max}^{+})^2)$ and $O(\Delta_{max}^{+})$ respectively, since each operation increments/decrements the load of at most one vertex each.
Hence, we have that the total sum of $\Phi$ increases over the entire update sequence, $\Delta(\Phi)$, satisfies:
\[\Delta(\Phi) \leq O(i\cdot{}\gamma(\Delta_{max}^{+})^2 + d\cdot{}\gamma\Delta_{max}^{+})\]
Deleting an $\eta$-invalid edge and reinserting it again, decreases the out-degree of a vertex of degree, say $\delta$, and increases the out-degree of a vertex of degree $\leq \delta-\eta-1$. 
Hence, for the low degree vertex the potential decreases by at least
\[(\Delta_{max}^{+}-\delta+\eta+1)^2-(\Delta_{max}^{+}-\delta+\eta)^2 = 1+2\cdot{(\Delta_{max}^{+}-\delta+\eta})\]
For the high-degree vertex the potential is increased by at most: 
\[(\Delta_{max}^{+}-\delta+1)^2-(\Delta_{max}^{+}-\delta)^2 = 1+2\cdot{(\Delta_{max}^{+}-\delta})\]
Hence, in total the potential decreases by $2\eta$ units.
Therefore, we can reorient at most $O(\frac{\gamma\Delta_{max}^{+}}{\eta}(\Delta_{max}^{+}\cdot{}i+d))$ $\eta$-invalid edges. Each reorientation takes two updates - one insertion and one deletion.
Thus, we arrive at the stated bound.
\end{proof}
Lemma \ref{lma:potential} allows us to bound the total number of edges moved between pseudoforests:
\begin{lemma}
\label{lma:pseudoMoves}
We move at most $O(\frac{\gamma(\Delta_{max}^{+})^3}{\eta}(i+d))$ edges between pseudoforests.
\end{lemma}
\begin{proof}
Each time an edge has its load changed by the insertion/deletion algorithm, it might cause an edge to have its orientation changed, and hence end up in a new pseudo forest. 
This might cause the algorithm from Lemma \ref{lma:split} to move $O(1)$ other edges between pseudoforests, and consequently these edges may also end up in the surplus graph.
Since each insertion/deletion changes the fractional orientation of at most $\Delta^{+}$ edges, the Lemma follows from Lemma \ref{lma:potential}.
\end{proof}
Note that this implies that the total no.\ of insertions into $R$ is  $O(\frac{\gamma(\Delta_{max}^{+})^3}{\eta}(i+d))$.
 
\subsection{Locating singly $\eta$-valid edges}
\label{sct:locating}
When we are accessing an edge, we can check if it is doubly $\eta$-valid or not (this information depends only on the load on the endpoints), and maintain this information in a dynamic forest using just 1 bit of information per edge. 
This allows us to later locate these edges using non-local searches in top trees.
However, when edges go between being singly $\eta$-valid and doubly $\eta$-valid through operations not accessing said edge, we are not able to maintain this information.
This can happen in two ways: either 1) a vertex has its load lowered causing an in-going edge to now become doubly $\eta$-valid or an outgoing edge to become singly $\eta$-valid or 2) a vertex has its load increased causing similar issues. 
We say an edge is \emph{clean} if we updated the validity bit of an edge, the last time the out-degree of an endpoint of the edge was altered. Otherwise, we say it is \emph{dirty}. 
Now if all edges on a cycle are clean, we can use top trees to direct searches for the edges that become invalidated by inverting the cycle.

We maintain a heavy-light decomposition of every forest using dynamic st-trees \cite{10.1145/800076.802464} to help us ensure that we can clean all edges in a cycle in time $O(\log^2{n})$. 
The idea is to maintain the invariant that all heavy edges are clean. 
Now we can clean a cycle by cleaning the at most $O(\log{n})$ light edges on said cycle.
In order to realise this invariant, whenever the degree of a vertex is changed, we need to update all of its incident heavy edges in all of the heavy-light decompositions.
Since a vertex is incident to at most 2 heavy edges in each forest, we have to update $O(\alpha)$ heavy edges. 
The following holds:
\begin{observation}
\label{obs:topsearch}
We can locate singly $\eta$-valid edges on a clean cycle in time $O(\log{n})$ per edge, if we spend $O(\log{n})$ overhead updating the bit indicating double validity.
\end{observation}
\begin{proof}
With $O(\log{n})$ overhead, we can update a top tree with this information, so that we can locate singly valid edges in time $O(\log{n})$ per singly valid edge using non-local searches \cite{10.1145/1103963.1103966}.
\end{proof}
\begin{observation}
\label{obs:heavylight}
We can insert and delete edges in the heavy-light decomposition in worst case $O(\log^2{n})$ time.
\end{observation}
\begin{proof}
We maintain the heavy-light decomposition using dynamic st-trees \cite{10.1145/800076.802464}. 
An insertion (or a deletion) creates at most $O(\log{n})$ new heavy edges. By Observation \ref{obs:topsearch}, we have to pay $O(\log{n})$ overhead to clean each edge.
\end{proof}
\begin{lemma}
\label{lma:clean}
We can check if a cycle is doubly valid in time $O(\log^2{n})$.
\end{lemma}
\begin{proof}
The algorithm maintains the invariant that all heavy edges are clean, hence the only dirty edges are light edges. 
There are at most $\log{n}$ such edges on any root-to-leaf path, and so at most $2\log{n}$ such edges on the path between any two vertices in the tree. 
We can locate each light edge in time $O(1)$ per edge. 
Indeed, we begin at $u$ resp.\ $v$ and chase pointers to the end of heavy paths on the way to the root. 
Now we clean each light edge in time $O(\log{n})$ by Observation \ref{obs:topsearch}.
Finally, we can search for a singly valid edge in time $O(\log{n})$ using Observation \ref{obs:topsearch}. If no edge is returned, we conclude that the cycle is doubly valid.
\end{proof}
Finally, we can show Lemma \ref{lma:surplusinsertion} (we recall it below):
\begin{lemmaS}[Lemma \ref{lma:surplusinsertion}]
After inserting an edge into $G[M]$, we can restore acyclicity and colourfulness in $O(\alpha^{3}\log^{2}(n))$ time.
\end{lemmaS}
\begin{proof}
We begin by showing the following two observations:
\begin{observation}
\label{obs:locateswitch}
We can locate a chain of Lemma \ref{lma:move1} switches in time $O(\alpha)$, and a Lemma \ref{lma:move2} switch in time $O(\alpha^2\log{n})$. 
\end{observation}
\begin{proof}
A component $V_C$ has size at most $O(\alpha)$, so we can determine colourfulness and locate relevant paths and cycles using for example a DFS in $O(\alpha)$ time. 

To locate a Lemma \ref{lma:move2} switch for a vertex $v$ on the cycle, we check for each $i$ such that $M_i \cap V_C \neq \emptyset$ whether $N_{i}(v)\cap V_C$ is empty. 
We check an $i$ in time $O(\alpha\log{n})$ by checking whether $v$ and $u \in V_{c}$ are neighbours in $P_i$ by using the top trees.
Since we check at most $O(\alpha)$ different $i$'s, the statement follows.
\end{proof}
\begin{observation}
\label{obs:performswitch}
We can perform a Lemma \ref{lma:move1} switch in $O(\log^2{n})$-time.
\end{observation}
\begin{proof}
Before performing a switch we determine in $O(\log^2{n})$ time, using Lemma \ref{lma:clean}, if the cycles are doubly valid.
If a cycle is singly valid, we still achieve our goal by Section \ref{sct:nondoubly}. 
We ascribe the cost of fixing singly valid edges to the potential in Lemma \ref{lma:potential}. 
Finally, we can reorient a cycle by inverting it with $ x =  1-\delta$ using Lemma \ref{lma:toptrees} in $O(\log{n})$ time. This ensures that every edge on the cycle now prefers the other endpoint, and so is naively rounded to the opposite direction.
Also note that this ensures that no edge on the cycle ends up in $H$.
\end{proof} 
If we end up with a non-colourful component, by Observation \ref{obs:locateswitch} and \ref{obs:performswitch} we spend $O(\alpha \log^2{n})$ time locating a path and performing switches, since such a path has length at most $O(\alpha)$.
We possibly need to do this for every pair of edges belonging to the same pseudoforest in said component.
In the worst case, the inserted edge connects to colourful components. Thus there are at most 2 edges from the same pseudoforest, and therefore at most one pair per pseudoforest.
Hence, there are at most $O(\alpha)$ such pairs, and we can handle each pair in $O(\alpha \log^2{n})$ time by above. 

If we end up with a cyclic component, we spend time $O(\alpha^2\log{n})$ locating a Lemma \ref{lma:move2} switch.
In order to perform it, we do at most $O(\alpha)$ Lemma \ref{lma:move1} switches. 
We fail this process at most $O(\alpha)$ times, since each time we fail the size of the component in the surplus graph decreases. The component has size $O(\alpha)$, since the graph was colourful before the insertion. When the process is successful, we can swap the edge in $G[M]$ in $O(\log{n})$ time.
Finally, this might yield an even bigger non-colourful, but acyclic component that we handle as before.
Correctness follows since no switches can create any new cycles.
\end{proof}

\subsection{Conclusion}
Theorem \ref{ArbDecompMain} follows from Lemma \ref{lma:amortisedUpdates} below. We prove Corollary \ref{cor:ColouringMain} in Section \ref{sct:colouring}.
\begin{lemma}
\label{lma:amortisedUpdates}
Consider a sequence of updates with $i$ insertions and $d$ deletions.
\begin{enumerate} 
    \item The insertion/deletion algorithm spends $O(\log^{6}{(n)}\cdot{}\alpha_{max}^{4}\cdot{}\eps^{-12}(i+d))$ time to update the fractional out-degree orientation and the refinement.
    \item The algorithm maintaining the pseudoforests spends $O(\log^{6}{(n)}\cdot{}\alpha_{max}^{3}\cdot{}\eps^{-8}(i+d))$ time.
    \item The algorithm maintaining the surplus graph spends $O(\log^{6}{(n)}\cdot{}\alpha_{max}^{6}\cdot{}\eps^{-8}(i+d))$ time.
\end{enumerate}
\end{lemma}
\begin{proof}
\begin{enumerate}
    \item By Lemma \ref{lma:implicitAccess}, $|L| = O(\alpha)$ and $q(n) = O(\alpha \log{n})$. Furthermore, we spend time $O(\alpha \log^{2}{n})$ updating heavy-edges per vertex accessed per insertion/deletion.
    By Lemma \ref{lma:potential}, we perform at most $O(\frac{\gamma\Delta_{max}^{+}}{\eta}(\Delta^{+}\cdot{}i+d))$ updates and so setting $\gamma = \Theta(\log{n}/\eps^{-2})$, $\eta = 1$ and applying Theorem \ref{thm:RefinementKop} yields the result.
    \item We spend overhead $O(\log^2{n})$ for each such operation, and by Lemma \ref{lma:pseudoMoves} we perform at most $O(\frac{\gamma(\Delta_{max}^{+})^3}{\eta}(i+d))$ of these operations. Thus the claimed follows by inserting parameters.
    \item An edge can only end up in a surplus graph, if it inserted into a new pseudoforest, so the statement follows from Lemmas \ref{lma:pseudoMoves} and \ref{lma:surplusinsertion}.
\end{enumerate}
\end{proof}
Now we can show Theorem \ref{ArbDecompMain}. We recall it below:
\begin{theoremS}[Theorem \ref{ArbDecompMain}]
Given an initially empty and dynamic graph undergoing an arboricity $\alpha_{max}$ preserving sequence of updates, there exists an algorithm maintaining a $\lfloor (1+\eps)\alpha \rfloor+2$ arboricity decomposition with an amortized update time of $O(\operatorname{poly}(\log{n},\alpha_{max},\varepsilon^{-1}))$. 
In particular, setting $\eps < \alpha_{max}^{-1}$ yields $\alpha+2$ forests with an amortized update time of $O(\operatorname{poly}(\log{n},\alpha_{max}))$. 
\end{theoremS}
\begin{proof}
Apply the algorithm from Theorem \ref{thm:RefinementKop} with $\eps' = \eps/10$,$\gamma = \log{n}/\eps'^2$, $\delta = 2\gamma^{-1}$ and $\mu = \gamma^{-1}$ as in the proof of Theorem \ref{OutOrientationMain} in order to insert/delete edges and maintain the refinements. 

Now Lemma \ref{lma:amortisedUpdates} shows that the amortised cost of running the algorithm from Section \ref{sec:DynFo} is $O(\log^{6}{(n)}\cdot{}\alpha_{max}^{6}\cdot{}\eps^{-12})$. 
\end{proof}

We have shown how to maintain an $\alpha+2$ arboricity decomposition of a fully dynamic graph as it undergoes an arboricity $\alpha$ preserving sequence of updates in $\poly(\log{n},\alpha)$ time per update.
We have also shown how to maintain an $\lfloor(1+\eps)\alpha\rfloor+2$ out-orientation of a fully dynamic graph in $\poly(\log{n},\alpha)$ time per update.
These algorithms are the first dynamic algorithms to go below $2\alpha$ forests and out-edges, respectively, and the number of forests matches the best near-linear static algorithm by Blumenstock and Fischer \cite{blumenstock2019constructive}. 
We apply these algorithms to get new trade-offs for implicit colouring algorithms for bounded arboricity graphs. 
In particular, we maintain $4\cdot{}2^{\alpha}$ and $2\cdot{}3^{\alpha}$ implicit colourings in $\poly(\log{n},\alpha)$ time per update. This improves upon the $2^{40\alpha}$ colours of the previous most colour-efficient algorithm maintaining $\poly(\log{n},\alpha)$ update time \cite{henzinger2020explicit}. 
In particular, this reduces the number of colours for planar graphs from $2^{120}$ to $32$.
An interesting direction for future work is to see, if one can reduce the number of forests even further in the static case, while still achieving near-linear running time. 
Also, even though our algorithms use few colours and forests, the update times contain quite high polynomials in both $\log{n}$ and $\alpha$. Is it possible to get more efficient update times without using more forests?
Finally, for constant $\alpha$, we get $\alpha+2$ out-edges.
Brodal \& Fagerberg \cite{Brodal99dynamicrepresentations} showed that one cannot get $\alpha$ out-edges with faster than $\Omega(n)$ update time (even amortised). 
The question remains, can one get $\alpha+1$?

\section{Acyclic Orientations and Arboricity Decompositions}
\label{ACYCLIC}
\label{sec:BroFa}
In this section, we describe the algorithm in Theorem \ref{thm:ModificationBroFa}.
We modify an algorithm by Brodal \& Fagerberg \cite{Brodal99dynamicrepresentations}. Specifically, we change how an edge is inserted.
The algorithm maintains a list of out-edges $\text{out}(u)$ for each vertex $u\in G$. An edge $e$ is in $\text{out}(u)$ if and only if $e \in E(G)$ and $e$ is oriented away from $u$. 
As a result $d^{+}(u) = |\text{out}(u)|$.
All of these lists are initialized to be empty. 
The algorithm ensures that the maximum out-degree of the vertices in $G$ is $d$ for some constant $d > 2\alpha$ to be specified later.
The algorithm handles deletions and insertions in the following way:

\noindent\textbf{Deletion:} If $e$ incident to $x,y$ 
is deleted, we search $\text{out}(x)$ and $\text{out}(y)$ for $e$, and delete it.

\noindent\textbf{Insertion:} When an edge $e$ is inserted, an arbitrary endpoint $u$ of $e$ is chosen, and $e$ is added to $\text{out}(u)$. 
Now every edge in $\text{out}(u)$ is oriented in the other direction (also $e$) i.e. we delete $f = (u,v)$ from $\text{out}(u)$ and add $f$ to $\text{out}(v)$ instead for all edges $f \in \text{out}(u)$. 
Now $u$ has out-degree at most $d$, but the reorientations of an edge $f = (u,v)$ might increase $|\text{out}(v)|$ above $d$. 
The algorithm then proceeds by reorienting all out-edges out of $v$.
It continues this process until all vertices have out-degree at most $d$.

Note that this process terminates: Since $G$ has arboricity $\alpha$, it has an orientation $O$ such that the maximum out-degree in $G$ is $\alpha$.
Call an edge in $E(G)$ \textit{good}, if it is oriented the same way by both the algorithm and $O$ and \textit{bad} if it isn't. 
Now, inserting $e$ could, in the worst case, make the algorithm change the orientation of $\alpha$ good edges.
However from here on, every vertex whose edges are reoriented will increase the total number of good edges by at least $d-2\alpha \geq 1$, so the process terminates. 
The algorithm differs from the one presented in \cite{Brodal99dynamicrepresentations} in only one way. When an edge $e = uv$ is inserted, we always turn $u$ into a sink. In \cite{Brodal99dynamicrepresentations}, this only happens if $u$'s out-degree increases above $d$.
This small modification ensures no cycles are created: when an edge is inserted, one of its endpoints is turned into a sink, and so this edge is in no cycle. Also, turning a vertex into a sink does not create any cycles.

\subsubsection{Correctness}
The algorithm clearly maintains a $d$-bounded out-degree orientation. 
Hence, correctness of the algorithm follows by showing Lemma \ref{lma:acyclicBF}:
\begin{lemma}
\label{lma:acyclicBF}
Let $G_s$ be $G$ after the first $s$ updates.
Then the orientation $O_s$ maintained by the algorithm is acyclic.
\end{lemma}
\begin{proof}
We will conceptually attach update time-stamps to all vertices:
Assume we run the algorithm on an initially empty and dynamic graph $G$ undergoing an arboricity $\alpha$ preserving sequence of insertions and deletions. 
After the first $s$ updates, the algorithm will have reoriented the out-edges out of some vertices.
Suppose $v^{s}_{1}, \dots , v^{s}_{k}$ is the order in which this is done.
That is for all $i$ the out-edges out of $v^{s}_{i}$ were reoriented by the algorithm before the out-edges out of $v^{s}_{i+1}$ were reoriented. 
After the first $s$ updates, we conceptually attach a time-stamp $t^{s}_v$ to each vertex $v \in V(G)$ equal to the maximum index $t$ such that $v^{s}_{t} = v$.
If $v \not \in v^{s}_{1}, \dots , v^{s}_{k}$, we define $t^{s}_{v} = -1$. 
\begin{claim}
\label{clm:timestamp}
Let $G_s$ be $G$ after the first $s$ updates.
Suppose $e = uv \in E(G_s)$ is oriented from $u$ towards $v$ in the orientation of $G_s$ maintained by the algorithm. 
Then $t^{s}_u < t^{s}_v$, that is the edges out of $u$ were last reoriented, before the edges out of $v$ were reoriented.
\end{claim}
\begin{proof}
Note that $t^{s}_v>-1$. Indeed, an edge cannot be oriented towards a vertex without that vertex having had its out-edges reoriented. 
Since all time-stamps except for $-1$ are unique, we cannot have $t^{s}_u = t^{s}_v$.
Suppose for contradiction that $t^{s}_u > t^{s}_v$. 
Since $e = uv$ only can have its orientation changed if either $u$ or $v$ has the orientation of all its out-edges flipped, the last time that the orientation of $e$ could have changed, was at time-step $t^{s}_u$.
But then $e$ cannot point away from $u$, since $u$ has out-degree $0$ after all its out-edges are reoriented.
\end{proof} 
Suppose now that there is a directed cycle $C = v_{0},v_{1}, \dots, v_{c}, v_{0}$ in $G_s$, with edges oriented from $v_{i}$ to $v_{i+1}$ for $i = 0, \dots, c-1$ and from $v_{c}$ towards $v_{0}$. 
Then by Claim \ref{clm:timestamp}, we must have that $t^{s}_{v_i} < t^{s}_{v_{i+1}}$ for all $i=1, \dots, c-1$. 
But by Claim \ref{clm:timestamp}, we also have $t^{s}_{v_c} < t^{s}_{v_{0}}$, and so we arrive at the contradiction $t^{s}_{v_{0}}<t^{s}_{v_c} < t^{s}_{v_{0}}$.
\end{proof}

\subsubsection{Analysis}
The analysis of the algorithm is synchronous to that in \cite{Brodal99dynamicrepresentations}: 
we just need to show a modified version of Lemma 1 in \cite{Brodal99dynamicrepresentations}. This reduces the problem to that of describing an offline reorientation scheme using few reorientations throughout the updates.
Here, we use the scheme presented as Lemma 3 in \cite{Brodal99dynamicrepresentations}. 
The idea behind this scheme is the following: 
If a vertex $v$ has out-degree $\delta>\alpha$, then it must have a neighbour, reachable by following only $O(\log_{\delta/\alpha}(n))$ out-edges, with out-degree no more than $\alpha$. 
Indeed, if the $t$-neighbourhood of $v$ does not contain such a vertex, then induction and the Theorem of Nash-Williams \cite{10.1145/800076.802464} implies that the $(t+1)$-neighbourhood must have size at least $(\delta/\alpha)^{t+1}$, and so the $O(\log_{\delta/\alpha}(n))$-neighbourhood must contain such a vertex. 
The reorientation scheme is then to accommodate insertion of an out-edge of $u$ by lowering the out-degree of $u$ via reorienting edges along such short paths. 
Since an edge must be inserted before it can be deleted, the reorientation scheme is run on the update sequence played in reverse in order to push the update cost unto delete operations, and we end up with a scheme achieving out-degree $\leq \delta$ using $r = c \cdot{} \lceil \log_{\delta /\alpha}{n} \rceil$ reorientations with $c$ being the number of deletes performed.
The modified version of the reduction to the existence of such a reorientation scheme is as follows.
\begin{lemma}
\label{lma:reductioalaBF}
Assume we are given an initially empty and dynamic graph $G$ undergoing an arboricity $\alpha$ preserving sequence $\sigma$ of insertions and deletions. Suppose there are $i$ insertions in $\sigma$.
If there exists a sequence of $\delta$-oriented graphs $\hat{G}_{1}, \dots \hat{G}_{|\sigma|} $ using at most $r$ reorientations between them in total, then the algorithm performs at most 
\[
(\delta\cdot{}i+r)(\frac{d+1}{d+1-2\delta})
\]
reorientations in total.
\end{lemma}
\begin{proof} The proof is synchronous to that of Lemma 2 in \cite{Brodal99dynamicrepresentations}.
We count how many times we can possibly turn a good edge bad, and then use this to bound the number of edges that are reoriented.
We do this by comparing the orientation of $G_j$ with that of $\hat{G}_{j}$.
We say an edge is \emph{bad}, if it is oriented differently in $G_j$ and $\hat{G}_{j}$.
Otherwise, we call it \emph{good}.
Doing this, we get a non-negative potential function
\[\Phi(G_{j}) = \text{the no. of bad edges in } E(G_{j})\]
Initially, there are no bad edges, since we begin with an empty graph. Furthermore, $\Phi$ is increased by at most $\delta$ for each of the $i$ edge insertions and 1 for each of the $r$ reorientations. Indeed, when we reorient all out-edges, we create at most $\delta$ bad edges.
Thus we lose at least $d+1-2\delta \geq 1$ bad edges every time the out-edges out of an out-degree $d+1$ vertex are reoriented.
Consequently, we can perform no more than 
\[
\frac{(\delta\cdot{}i+r)}{d+1-2\delta}
\]
such operations. 
Each such operation turns at most $\delta$ good edges bad, so the total number of times these operations turn a good edge bad is upper bounded by
\[
\frac{\delta(\delta\cdot{}i+r)}{d+1-2\delta}
\]
In particular, the number of bad edges the operation can turn good is at most $\delta i+r+\frac{\delta(\delta\cdot{}i+r)}{d+1-2\delta}$. 
Adding these numbers together, we see that the algorithm performs at most 
\[
\delta\cdot{}i+r+\frac{2\delta(\delta\cdot{}i+r)}{d+1-2\delta} = (\delta\cdot{}i+r)(1+\frac{2\delta}{d+1-2\delta}) = (\delta\cdot{}i+r)(\frac{d+1}{d+1-2\delta})
\]
reorientations, since each reorientation either turns a bad edge good or a good edge bad. 
\end{proof} 
Combining the scheme from \cite{Brodal99dynamicrepresentations} with Lemma \ref{lma:acyclicBF} yields the following for which specifying parameters yields Theorem \ref{thm:ModificationBroFa}:
\begin{theorem}
Given an initially empty and dynamic graph $G$ undergoing an arboricity $\alpha$ preserving sequence $\sigma$ of insertions and deletions (say $i$ insertions and $c$ deletions), the algorithm with parameters $d/2 \geq \delta > \alpha$ will have an amortized insertion cost of $O(\delta\cdot{}\frac{d+1}{d+1-2\delta})$, and an amortized deletion cost of $O(d+\log_{\frac{\delta}{\alpha}}(n)\frac{d+1}{d+1-2\delta})$.
\end{theorem}
\begin{proof}
For any $\delta > \alpha$ the reorientation scheme from \cite{Brodal99dynamicrepresentations} gives $\delta$-orientations $O_{j}$ such that $O_{j}$ is a $\delta$-bounded out-degree orientation of $G_{j}$ for all $j = 1, \dots |\sigma|$ such that the total number of reorientations is at most $ \lceil \log_{\frac{\delta}{\alpha}}(n)) \rceil$ times the number of deletions.
Each reorientation can be done in $O(1)$ time per reorientation, since we reorient the entire out-degree list.
For deletions, we incur an extra cost of $O(d)$ to locate the deleted edge in the out-list.

Applying Lemma \ref{lma:reductioalaBF}, we find that the algorithm performs at most $(\delta \cdot{i}+c\cdot{}\lceil\log_{\frac{\delta}{\alpha}}(n)\rceil)\frac{d+1}{d+1-2\delta}$ reorientations in total. 
As such we retrieve the stated update costs.
\end{proof} 
Setting $\delta = \alpha_{max}+1$ and $d = 2\delta$ in the above Theorem, we arrive at Theorem \ref{thm:ModificationBroFa}.
\section{Dynamic Colouring}
\label{sct:colouring}
In this section, we discuss how to turn dynamic arboricity decomposition algorithms and dynamic out-orientation algorithms into implicit dynamic colouring algorithms, and use this to deduce Corollaries \ref{cor:ColouringSecondary} and \ref{cor:ColouringMain}.
For the dynamic arbroricity decomposition algorithms, we do exactly as in \cite{henzinger2020explicit}. 
Here, Henzinger et al.\ maintain each forest in a data structure for maintaining information in dynamic forests, more specifically they use top trees \cite{10.1145/1103963.1103966}. 
Now they root each forest arbitrarily and colour the forests by the parity of the distance to the root. 
Finally, in order to obtain a proper colouring of the whole graph, they return the product colouring over all of the forests in the arboricity decomposition. 
It costs $O(\log{n})$ overhead to maintain each forest as a top forest, and so they arrive at 
\begin{lemma}[\cite{henzinger2020explicit}]
\label{lma:ColourForest}
Given blackbox access to a dynamic algorithm maintaining an $\alpha'$ arboricity decomposition of a graph $G$ with update time $T(n)$, there exists a dynamic algorithm maintaining an implicit $2^{\alpha'}$ colouring of $G$ with update time $O(T(n)\log{n})$ and query-time $O(\alpha' \log{n})$.
\end{lemma}
In order to use dynamic out-orientations, they show how to turn a dynamic algorithm maintaining an $\alpha'$-bounded out-degree orientation into a dynamic algorithm maintaining a $2\alpha'$-arboricity decomposition. Finally, they apply Lemma \ref{lma:ColourForest} to get a $2^{2\alpha'} = 4^{\alpha'}$ colouring. 

We use the following strategy instead. We use Lemma \ref{lma:split} to turn the out-orientation into a pseudoarboricity decomposition.  
Then we maintain each pseudoforest as a forest and a matching. Finally, we root each pseudotree in an arbitrary vertex on the unicycle (as we did in Section \ref{sct:recovery}).
Now with $O(\log{n})$ overhead, we can 3 colour each pseudoforest. 
Indeed, we store each pseudotree as a rooted top tree plus an edge such that one of the endpoints of the edge not in the top tree is the root of the top tree.
Then we can colour the root a unique colour and all other vertices by the parity of the distance to the root as before. 
Finally, we form the final colouring as the product colouring over all pseudoforests. This colouring can again be queried in time $O(\alpha' \log{n})$.

Using these strategies for colouring, we can now deduce Corollaries \ref{cor:ColouringMain} and \ref{cor:ColouringSecondary}.
\begin{corollaryS}[Corollary \ref{cor:ColouringMain}]
Given a dynamic graph with $n$ vertices, there exists a fully dynamic algorithm that maintains an implicit $4\cdot{}2^{\alpha}$ colouring with an amortized update of $O(\operatorname{polylog}{n})$ and a query time of $O(\alpha\log{n})$.
\end{corollaryS}
\begin{proof}
We schedule the updates on $O(\log{n})$ copies of $G$ using Theorem \ref{thm:schedule}. 
On each copy we apply the algorithm of Theorem \ref{ArbDecompMain} with $\alpha_{max} = 4\rho_{est}$ and $\eps < \alpha_{max}^{-1}$ combined with Lemma \ref{lma:ColourForest}. In the remaining copies, we do nothing. 
In order to answer a colouring query, we query the currently active copy. 
If a copy with $\rho_{est} > \log{n}$ is active, we just colour each vertex by their vertex id, since then $\alpha(G) \geq \rho_{est} \geq \log{n}$, and so $2^{\alpha(G)} \geq n$.
\end{proof}
\begin{corollaryS}[Corollary\ \ref{cor:ColouringSecondary}]
Given a dynamic graph with $n$ vertices, there exists a fully dynamic algorithm that maintains an implicit $2\cdot{}3^{(1+\varepsilon)\alpha}$ colouring with an amortized update time of $O(\log^{4}{n}\cdot{}\alpha^2/\eps^6)$ and a query time of $O(\alpha\log{n})$.
\end{corollaryS}
\begin{proof}
Apply the algorithm from Theorem \ref{thm:RefinementKop} with $\eps' = \eps/20$,$\gamma = \log{n}/\eps'^2$, $\delta = 2\gamma^{-1}$ and $\mu = \gamma^{-1}$ as in the proof of Theorem \ref{OutOrientationMain}. 
Naively round each edge in $G-H$ to get a forest and a $\lfloor (1+\eps)\alpha' \rfloor$ out-orientation. 
Apply Lemma \ref{lma:split} to get a pseudoforest decomposition of $G-H$. Colour $G-H$ with the strategy for colouring pseudoforests discussed above, and colour $H$ with the strategy from Lemma \ref{lma:ColourForest}. 
Finally, take the product colouring.
\end{proof}

\bibliographystyle{plain}
\bibliography{bibliography}

\begin{thebibliography}{10}

\bibitem{db16bc9212da408a92161ab5d03b6a88}
Oswin Aichholzer, Franz Aurenhammer, and G{\"u}nter Rote.
\newblock {\em Optimal graph orientation with storage applications}.
\newblock SFB-Report , SFB 'Optimierung und Kontrolle'. 1995.
\newblock Reportnr.: F003-51.

\bibitem{10.1145/1103963.1103966}
Stephen Alstrup, Jacob Holm, Kristian~De Lichtenberg, and Mikkel Thorup.
\newblock Maintaining information in fully dynamic trees with top trees.
\newblock {\em ACM Trans. Algorithms}, 1(2):243–264, October 2005.

\bibitem{AppelH76}
K.~Appel and W.~Haken.
\newblock A proof of the four color theorem.
\newblock {\em Discret. Math.}, 16(2):179--180, 1976.

\bibitem{DBLP:journals/dam/ArikatiMZ97}
Srinivasa~Rao Arikati, Anil Maheshwari, and Christos~D. Zaroliagis.
\newblock Efficient computation of implicit representations of sparse graphs.
\newblock {\em Discret. Appl. Math.}, 78(1-3):1--16, 1997.

\bibitem{Banerjee}
Niranka Banerjee, Venkatesh Raman, and Saket Saurabh.
\newblock Fully dynamic arboricity maintenance.
\newblock In {\em Computing and Combinatorics - 25th International Conference,
  {COCOON} 2019, Xi'an, China, July 29-31, 2019, Proceedings}, volume 11653 of
  {\em Lecture Notes in Computer Science}, pages 1--12. Springer, 2019.

\bibitem{Barba}
L.~Barba, J.~Cardinal, M.~Korman, S.~Langerman, A.~van Renssen, M.~Roeloffzen,
  and S.~Verdonschot.
\newblock Dynamic graph coloring.
\newblock {\em Algorithmica}, 81(4):1319--1341, 2019.

\bibitem{Elkin}
Leonid Barenboim and Michael Elkin.
\newblock Sublogarithmic distributed mis algorithm for sparse graphs using
  nash-williams decomposition.
\newblock In {\em Proceedings of the Twenty-Seventh ACM Symposium on Principles
  of Distributed Computing}, PODC '08, page 25–34, New York, NY, USA, 2008.
  Association for Computing Machinery.

\bibitem{berglinetal:LIPIcs:2017:8263}
Edvin Berglin and Gerth~Stolting Brodal.
\newblock {A Simple Greedy Algorithm for Dynamic Graph Orientation}.
\newblock In {\em 28th International Symposium on Algorithms and Computation
  (ISAAC 2017)}, volume~92 of {\em Leibniz International Proceedings in
  Informatics (LIPIcs)}, pages 12:1--12:12, Dagstuhl, Germany, 2017. Schloss
  Dagstuhl--Leibniz-Zentrum fuer Informatik.

\bibitem{DBLP:conf/soda/BhattacharyaCHN18}
Sayan Bhattacharya, Deeparnab Chakrabarty, Monika Henzinger, and Danupon
  Nanongkai.
\newblock Dynamic algorithms for graph coloring.
\newblock In {\em Proceedings of the Twenty-Ninth Annual {ACM-SIAM} Symposium
  on Discrete Algorithms, {SODA} 2018, New Orleans, LA, USA, January 7-10,
  2018}, pages 1--20. {SIAM}, 2018.

\bibitem{DBLP:journals/corr/abs-1910-02063}
Sayan Bhattacharya, Fabrizio Grandoni, Janardhan Kulkarni, Quanquan~C. Liu, and
  Shay Solomon.
\newblock Fully dynamic ({\(\Delta\)}+1)-coloring in constant update time.
\newblock {\em CoRR}, abs/1910.02063, 2019.

\bibitem{doi:10.1137/1.9781611974317.10}
Markus Blumenstock.
\newblock Fast algorithms for pseudoarboricity.
\newblock In {\em Proceedings of the Eighteenth Workshop on Algorithm
  Engineering and Experiments, {ALENEX} 2016, Arlington, Virginia, USA, January
  10, 2016}, pages 113--126. {SIAM}, 2016.

\bibitem{blumenstock2019constructive}
Markus Blumenstock and Frank Fischer.
\newblock A constructive arboricity approximation scheme.
\newblock In {\em {SOFSEM} 2020: Theory and Practice of Computer Science - 46th
  International Conference on Current Trends in Theory and Practice of
  Informatics, {SOFSEM} 2020, Limassol, Cyprus, January 20-24, 2020,
  Proceedings}, volume 12011 of {\em Lecture Notes in Computer Science}, pages
  51--63. Springer, 2020.

\bibitem{Brodal99dynamicrepresentations}
Gerth~Stolting Brodal and Rolf Fagerberg.
\newblock Dynamic representations of sparse graphs.
\newblock In {\em In Proc. 6th International Workshop on Algorithms and Data
  Structures (WADS)}, pages 342--351. Springer-Verlag, 1999.

\bibitem{ChibaNS81}
Norishige Chiba, Takao Nishizeki, and Nobuji Saito.
\newblock A linear 5-coloring algorithm of planar graphs.
\newblock {\em J. Algorithms}, 2(4):317--327, 1981.

\bibitem{RealLifeCol}
O.~{Coudert}.
\newblock Exact coloring of real-life graphs is easy.
\newblock {\em Proceedings of 34th Design Automation Conference. ACM},
  35(1):121–126, 1997.

\bibitem{Edmonds1965MinimumPO}
Jack Edmonds.
\newblock Minimum partition of a matroid into independent subsets.
\newblock {\em Journal of Research of the National Bureau of Standards Section
  B Mathematics and Mathematical Physics}, page~67, 1965.

\bibitem{10.1016/0020-0190(94)90121-X}
David Eppstein.
\newblock Arboricity and bipartite subgraph listing algorithms.
\newblock {\em Inf. Process. Lett.}, 51(4):207–211, August 1994.

\bibitem{Frederickson84}
Greg~N. Frederickson.
\newblock On linear-time algorithms for five-coloring planar graphs.
\newblock {\em Inf. Process. Lett.}, 19(5):219--224, 1984.

\bibitem{10.1145/62212.62252}
Harold Gabow and Herbert Westermann.
\newblock Forests, frames, and games: Algorithms for matroid sums and
  applications.
\newblock In {\em Proceedings of the Twentieth Annual ACM Symposium on Theory
  of Computing}, STOC '88, page 407–421, New York, NY, USA, 1988. Association
  for Computing Machinery.

\bibitem{10.5555/313651.313667}
Harold~N. Gabow.
\newblock Algorithms for graphic polymatroids and parametric s-sets.
\newblock In {\em Proceedings of the Sixth Annual ACM-SIAM Symposium on
  Discrete Algorithms}, SODA '95, page 88–97, USA, 1995. Society for
  Industrial and Applied Mathematics.

\bibitem{Ghaffari}
Mohsen Ghaffari and Hsin-Hao Su.
\newblock {\em Distributed Degree Splitting, Edge Coloring, and Orientations},
  pages 2505--2523.

\bibitem{Harris2}
David~G. Harris.
\newblock Distributed local approximation algorithms for maximum matching in
  graphs and hypergraphs.
\newblock In {\em 2019 IEEE 60th Annual Symposium on Foundations of Computer
  Science (FOCS)}, pages 700--724, 2019.

\bibitem{Harrisetal}
David~G. Harris, Hsin-Hao Su, and Hoa~T. Vu.
\newblock On the locality of nash-williams forest decomposition and star-forest
  decomposition.
\newblock In {\em Proceedings of the 2021 ACM Symposium on Principles of
  Distributed Computing}, PODC'21, page 295–305, New York, NY, USA, 2021.
  Association for Computing Machinery.

\bibitem{He2014OrientingDG}
Meng He, Ganggui Tang, and N.~Zeh.
\newblock Orienting dynamic graphs, with applications to maximal matchings and
  adjacency queries.
\newblock In {\em ISAAC}, 2014.

\bibitem{henzinger2020explicit}
Monika Henzinger, Stefan Neumann, and Andreas Wiese.
\newblock Explicit and implicit dynamic coloring of graphs with bounded
  arboricity.
\newblock {\em CoRR}, abs/2002.10142, 2020.

\bibitem{DBLP:conf/stacs/Henzinger020}
Monika Henzinger and Pan Peng.
\newblock Constant-time dynamic ({\(\Delta\)}+1)-coloring.
\newblock In {\em 37th International Symposium on Theoretical Aspects of
  Computer Science, {STACS} 2020, March 10-13, 2020, Montpellier, France},
  volume 154 of {\em LIPIcs}, pages 53:1--53:18. Schloss Dagstuhl -
  Leibniz-Zentrum f{\"{u}}r Informatik, 2020.

\bibitem{DBLP:conf/icalp/KhotP06}
Subhash Khot and Ashok~Kumar Ponnuswami.
\newblock Better inapproximability results for maxclique, chromatic number and
  min-3lin-deletion.
\newblock In {\em Automata, Languages and Programming, 33rd International
  Colloquium, {ICALP} 2006, Venice, Italy, July 10-14, 2006, Proceedings, Part
  {I}}, volume 4051 of {\em Lecture Notes in Computer Science}, pages 226--237.
  Springer, 2006.

\bibitem{kopelowitz2013orienting}
Tsvi Kopelowitz, Robert Krauthgamer, Ely Porat, and Shay Solomon.
\newblock Orienting fully dynamic graphs with worst-case time bounds.
\newblock In {\em Automata, Languages, and Programming - 41st International
  Colloquium, {ICALP} 2014, Copenhagen, Denmark, July 8-11, 2014, Proceedings,
  Part {II}}, volume 8573 of {\em Lecture Notes in Computer Science}, pages
  532--543. Springer, 2014.

\bibitem{10.1007/11940128_56}
\L{}ukasz Kowalik.
\newblock Approximation scheme for lowest outdegree orientation and graph
  density measures.
\newblock In {\em Proceedings of the 17th International Conference on
  Algorithms and Computation}, ISAAC'06, page 557–566, Berlin, Heidelberg,
  2006. Springer-Verlag.

\bibitem{10.1016/j.ipl.2006.12.006}
\L{}ukasz Kowalik.
\newblock Adjacency queries in dynamic sparse graphs.
\newblock {\em Inf. Process. Lett.}, 102(5):191–195, May 2007.

\bibitem{LeeSidford}
Yin~Tat Lee and Aaron Sidford.
\newblock Path finding methods for linear programming: Solving linear programs
  in Õ(vrank) iterations and faster algorithms for maximum flow.
\newblock In {\em 2014 IEEE 55th Annual Symposium on Foundations of Computer
  Science}, pages 424--433, 2014.

\bibitem{Madry}
Aleksander Madry.
\newblock Navigating central path with electrical flows: From flows to
  matchings, and back.
\newblock In {\em 2013 IEEE 54th Annual Symposium on Foundations of Computer
  Science}, pages 253--262, 2013.

\bibitem{Matula1980TwoLA}
D.~Matula, Y.~Shiloach, and R.~Tarjan.
\newblock Two linear-time algorithms for five-coloring a planar graph.
\newblock 1980.

\bibitem{https://doi.org/10.1112/jlms/s1-39.1.12}
C.~St.J.~A. Nash-Williams.
\newblock Decomposition of finite graphs into forests.
\newblock {\em Journal of the London Mathematical Society}, s1-39(1):12--12,
  1964.

\bibitem{https://doi.org/10.1002/net.3230120206}
Jean-Claude Picard and Maurice Queyranne.
\newblock A network flow solution to some nonlinear 0-1 programming problems,
  with applications to graph theory.
\newblock {\em Networks}, 12(2):141--159, 1982.

\bibitem{RobertsonSST96}
Neil Robertson, Daniel~P. Sanders, Paul~D. Seymour, and Robin Thomas.
\newblock Efficiently four-coloring planar graphs.
\newblock In Gary~L. Miller, editor, {\em Proceedings of the Twenty-Eighth
  Annual {ACM} Symposium on the Theory of Computing, Philadelphia,
  Pennsylvania, USA, May 22-24, 1996}, pages 571--575. {ACM}, 1996.

\bibitem{DBLP:conf/stoc/SawlaniW20}
Saurabh Sawlani and Junxing Wang.
\newblock Near-optimal fully dynamic densest subgraph.
\newblock In {\em Proccedings of the 52nd Annual {ACM} {SIGACT} Symposium on
  Theory of Computing, {STOC} 2020, Chicago, IL, USA, June 22-26, 2020}, pages
  181--193. {ACM}, 2020.

\bibitem{10.1145/800076.802464}
Daniel~D. Sleator and Robert~Endre Tarjan.
\newblock A data structure for dynamic trees.
\newblock In {\em Proceedings of the Thirteenth Annual ACM Symposium on Theory
  of Computing}, STOC '81, page 114–122, New York, NY, USA, 1981. Association
  for Computing Machinery.

\bibitem{10.1145/3392724}
Shay Solomon and Nicole Wein.
\newblock Improved dynamic graph coloring.
\newblock {\em ACM Trans. Algorithms}, 16(3), June 2020.

\bibitem{suetal:LIPIcs:2020:13093}
Hsin-Hao Su and Hoa~T. Vu.
\newblock {Distributed Dense Subgraph Detection and Low Outdegree Orientation}.
\newblock In Hagit Attiya, editor, {\em 34th International Symposium on
  Distributed Computing (DISC 2020)}, volume 179 of {\em Leibniz International
  Proceedings in Informatics (LIPIcs)}, pages 15:1--15:18, Dagstuhl, Germany,
  2020. Schloss Dagstuhl--Leibniz-Zentrum f{\"u}r Informatik.

\bibitem{10.1145/1132516.1132612}
David Zuckerman.
\newblock Linear degree extractors and the inapproximability of max clique and
  chromatic number.
\newblock In {\em Proceedings of the Thirty-Eighth Annual ACM Symposium on
  Theory of Computing}, STOC '06, page 681–690, New York, NY, USA, 2006.
  Association for Computing Machinery.

\end{thebibliography}
\appendix
\section{Missing details of Section 2}
\label{apx:KopSW}
\begin{theoremS}[Theorem \ref{thm:Kop-main},implicit in \cite{kopelowitz2013orienting}]
Given implicit $(|L|,q(n))$ access to an orientation with $\max_{v} s(v) \leq \Delta^{+}$, there exists an algorithm that can insert and delete edges from the orientation without creating any new $1$-invalid edges.
The algorithm has worst-case insertion time of $O(\gamma\cdot{}\Delta^{+}(\Delta^{+}+\log(n)(|L|+1)+q(n))))$ and a worst case deletion time of $O(\gamma\cdot{}\Delta^{+}(\log(n)(|L|+1)+q(n))))$.
\end{theoremS}
We provide pseudo-code for the algorithm in Theorem \ref{thm:Kop-main}: 
We store the following global data structure.
\begin{itemize}
    \item $Edges:$ Balanced binary search tree sorted by edge id. Initially empty.
\end{itemize}
Each vertex $u$ stores the following data structures:
\begin{itemize}
    \item $s(u)$, initialized to $0$.
    \item $OutNbrs_{u}$: Balanced binary search tree containing out-neighbours of $u$ sorted by vertex id. Initially empty
    \item $InNbrs_{u}$: Max heap containing in-neighbours indexed using $s(v)$. Initially empty.
\end{itemize}
We update them using the operations in Algorithms \ref{alg:basic-ops}, \ref{alg:implicit-access} and \ref{alg:insert-delete}.\
\begin{algorithm} 
    \caption{Basic operations}
    \label{alg:basic-ops}
    \begin{algorithmic}
            \Function{add}{$e,u \rightarrow v$}
                \If{$e\in Edges$}
                    \State Update $|B_{e}^{u}| \leftarrow |B_{e}^{u}| + 1$
                    \State Update $|B_{e}| \leftarrow |B_{e}| + 1$
                \Else 
                    \State add $e$ to $Edges$ with keys 
                    \State $|B_{e}^{u}| = 1$, $|B_{e}^{v}| = 0$,    $|B_{e}| = 1$.
                \EndIf
                \If{$u \notin InNbrs_{v}$}
                    \State add $u$ to $InNbrs_{v}$ with key $s(u)$
                \EndIf
                \If{$v \notin OutNbrs_{u}$}
                    \State add $v$ to $OutNbrs_{u}$
                \EndIf
            \EndFunction
            \\
            \Function{remove}{$e,u \rightarrow v$}
                \If{$|B_{e}| = 1$}
                    \State remove $e$ from $Edges$
                \EndIf
                \If{$|B_{e}^{u}| = 1$} 
                    \State Remove $u$ from $InNbrs_{v}$
                    \State Remove $v$ from $OutNbrs_{u}$
                \EndIf
            \EndFunction
            \\
            \Function{TightOutNbr}{$u$}
                \For{$w \in OutNbrs_{u}$}
                    \If{$s(w)-s(u) \leq -\eta/2$}
                        \State \Call{return}{} w
                    \EndIf
                \EndFor
                \State \Call{return}{} null
            \EndFunction
            \\
            \Function{TightInNbr}{$u$}
                \State $w \leftarrow max(InNbrs_{u})$
                \If{$s(u)-s(w) \leq -\eta/2$}
                    \State \Call{return}{} w
                \EndIf
                \State \Call{return}{} null
            \EndFunction
            \\
            \Function{increment}{$u$}
            \State $s(u) \leftarrow s(u)+1$
            \For{$v \in OutNbrs_{u}$}
                \State Update $s(u)$ in $InNbrs_{v}$
            \EndFor
            \EndFunction
            \\
            \Function{decrement}{$u$}
            \State $s(u) \leftarrow s(u)-1$
            \For{$v \in OutNbrs_{u}$}
                \State Update $s(u)$ in $InNbrs_{v}$
            \EndFor
            \EndFunction
    \end{algorithmic}
\end{algorithm}
\begin{algorithm}
    \caption{Accessing the implicit orientation}
    \label{alg:implicit-access}
        \begin{algorithmic}
            \Function{access}{$e,uv$}
                \State Update $|B_{e}^{u}|,|B_{e}^{v}| \leftarrow$ \Call{AccessLoad}{$e$} in $Edges$
            \EndFunction
            \\
            \Function{change}{$e,uv$}
                \State $x \leftarrow |B_{e}^{u}|$
                \State $y \leftarrow |B_{e}^{v}|$
                \State \Call{UpdateLoad}{$u,v,x,y$}
            \EndFunction
            \\
            \Function{reorient}{$e,u\rightarrow v$}
                \State \Call{access}{$e,uv$}
                \State \Call{remove}{$e,u\rightarrow v$}
                \State \Call{add}{$e,v\rightarrow u$}
                \State \Call{change}{$e,uv$}
            \EndFunction
            \\
            \Function{UpdateNbrs}{$u$}
                \State $L \leftarrow$ \Call{NewNbrs}{$u$}
                \For{$uv \in L$}
                    \State \Call{access}{$uv$}
                    \For{$i=1$ to 2}
                        \If{$B_{uv}^{u} = 0$}
                            \State Remove $u$ from $InNbrs_{v}$
                            \State Remove $v$ from $OutNbrs_{u}$
                        \Else 
                            \If{$u \notin InNbrs_{v}$}
                                \State add $u$ to $InNbrs_{v}$ with key $s(u)$
                            \EndIf
                            \If{$v \notin OutNbrs_{u}$}
                                \State add $v$ to $OutNbrs_{u}$
                            \EndIf
                        \EndIf
                        \State $u,v \leftarrow v,u$
                    \EndFor
                \EndFor
            \EndFunction
        \end{algorithmic}
\end{algorithm}
\begin{algorithm}
    \caption{Inserting procedure: }
    \label{alg:insert-delete}
        \begin{algorithmic}
            \Function{InsertCopy}{$e,uv$}
                \If{$s(u)\leq s(v)$}
                    \State \Call{add}{$e,u\rightarrow v$}
                    \State $w \leftarrow u$
                \Else 
                    \State \Call{add}{$e,v\rightarrow u$}
                    \State $w \leftarrow v$
                \EndIf
                \State \Call{UpdateNbrs}{$w$}
                \While{\Call{TightOutNbrs}{$w$}$\neq$ null}
                    \State $w' \leftarrow$ \Call{TightOutNbrs}{$w$}
                    \State \Call{reorient}{$ww',w\rightarrow w'$}
                    \State $w \leftarrow w'$
                    \State \Call{UpdateNbrs}{$w$}
                \EndWhile
                \State \Call{increment}{$w$}
            \EndFunction
            \\
            \Function{GammaInsert}{$e,uv$}
                \For{1 to $\gamma$}
                    \State \Call{InsertCopy}{$e,uv$}
                \EndFor
            \EndFunction
            \\
            \Function{DeleteCopy}{$e,u\rightarrow v$}
                \State \Call{remove}{$e,u\rightarrow v$}
                \State $w \leftarrow u$
                \State \Call{UpdateNbrs}{$w$}
                \While{\Call{TightInNbrs}{$w$}$\neq$ null}
                    \State $w' \leftarrow$ \Call{TightInNbrs}{$w$}
                    \State \Call{reorient}{$w'w,w'\rightarrow w$}
                    \State \Call{UpdateNbrs}{$w'$}
                \EndWhile
                \State \Call{decrement}{$w'$}
            \EndFunction
            \\
            \Function{GammaDelete}{$e,uv$}
                \For{1 to $\gamma$}
                    \If{$|B_{e}^{u}| >0 $}
                        \State \Call{DeleteCopy}{$e,u\rightarrow v$}
                    \Else 
                        \State \Call{DeleteCopy}{$e,v\rightarrow u$}
                    \EndIf
                \EndFor
            \EndFunction
        \end{algorithmic}
\end{algorithm}
Theorem \ref{thm:Kop-main} now follows from the following two Lemmas:
\begin{lemma}[Implicit in \cite{kopelowitz2013orienting}]
\label{lma:insertKop}
We can insert a copy of an edge into $G^{\gamma}$ in time $O(\Delta^{+}\cdot{}(\Delta^{+}+|L|\log{n}+q(n)))$
\end{lemma}
\begin{proof}
For each vertex on the chain, we have to process the neighbourhood changes in time $O(|L|\cdot{}\log{n}+q(n))$, look at all out-neighbours in time $O(\Delta^{+})$ and access and update an out-going edge in time $O(\log{n})$.
By Remark \ref{rmk:chainlength}, the chain has at most $\Delta^{+}$ vertices.
\end{proof}
\begin{lemma}[Implicit in \cite{kopelowitz2013orienting}]
\label{lma:deleteKop}
We can delete a copy of an edge into $G^{\gamma}$ in time $O(\Delta^{+}\cdot{}(\log{n}+|L|\log{n}+q(n)))$
\end{lemma}
\begin{proof}
For each vertex on the chain, we have to process the neighbourhood changes in time $O(|L|\cdot{}\log{n}+q(n))$, extract the maximum from a max-heap, and access and update an in-going edge in time $O(\log{n})$.
By Remark \ref{rmk:chainlength}, the chain has at most $\Delta^{+}$ vertices.
\end{proof}

\newpage
\section{Scheduling updates}
\label{apx:schedule}
In this Appendix, we paraphrase the algorithms and ideas of Sawlani \& Wang. To be consistent, we use the terminology introduced in Section \ref{sct:SWKop}. 
In particular, we consider out-orientations instead of the symmetric problem of orienting the edges in order to provide guarantees on the number of in-edges for each vertex. 
We also make similar modifications to the algorithm in Lemma 3.6 in \cite{DBLP:conf/stoc/SawlaniW20}, as we did to the algorithm of Kopelowitz et al.\ ~\cite{kopelowitz2013orienting} in Section \ref{sct:SWKop}, and show that this algorithm can be used to maintain implicit approximate out-orientations in dense graphs with polylogarithmic in $n$ update time by using the techniques of Section \ref{sec:Frac}.

A key insight in \cite{DBLP:conf/stoc/SawlaniW20} is to allow $\eta$ to scale with an estimate of the maximum density. 
Combined with the following guarantee for locally $\eta$-stable orientations and Remark \ref{rmk:chainlength}, this reduces the length of maximal tight chains to $O(\gamma)$, which means that the update times become independent of $\alpha(G)$. 
\begin{theorem}[\cite{DBLP:conf/stoc/SawlaniW20} Theorem 3.1 and Corollary 3.2] 
\label{thm:SW-EtaValid}
If $O$ is a locally $\eta$-stable and $\alpha'$-bounded fractional out-degree orientation of a graph $G$, then for $\rho_{est}\leq \alpha' = \max_{v} s(v) \leq 2\rho_{est}$ it holds that:
\[
\left(1-4\sqrt{\eta \log{n}/\rho_{est}}\right) \alpha' \leq \rho(G) \leq \alpha(G)
\]
\end{theorem}
However, in order to use this guarantee, one needs an estimate of the maximum density. Sawlani and Wang overcome this by maintaining $O(\log{n})$ copies of $G$ - each with a different estimate of $\rho$. Then there will always be a copy with the guarantees of Theorem \ref{thm:SW-Main}, provided of course that the algorithm is fully updated.

Furthermore, since $\eta$ now scales with $\rho$, we cannot afford to look at all out-neighbours, as we do during insertions in Theorem \ref{thm:Kop-main}. 
In order to overcome this, Sawlani and Wang lazily update different parts of the out neighbourhood each time a vertex' load is changed. Since a vertex's load now can be increased multiple times, before a tight edge becomes invalid, one is always able to identify an edge as tight, before it becomes invalid. We have the following Theorem:
\begin{theorem}[Implicit in \cite{DBLP:conf/stoc/SawlaniW20}]
\label{thm:SW-Main}
Given implicit $(|L|,q(n))$ access to an orientation with $\max_{v} s(v) \leq \rho_{max}$ where $\rho_{max}$ is an upper bound on the maximum density, there exists an algorithm that can insert and delete edges from the orientation without creating any new $ \left(\frac{\eps^2\cdot{}\rho_{max}}{16\log{n}}\right)$-invalid edges.
The algorithm has worst-case insertion and deletion time $O(\eps^{-6}\log^3{(n)}(1+\log(n)|L|+q(n)))$
\end{theorem}
Before we give a proof, we note that by using Theorem~\ref{thm:SW-Main} for updates and Theorem~\ref{thm:schedule} to schedule the updates in the proof of Theorem \ref{thm:RefinementKop}, we obtain:~
\begin{theorem}
\label{thm:RefinementSW}
Suppose $\gamma = 128 \log{n}/\eps^2$ and that $1>\delta>\gamma^{-1}+\mu> 2\gamma^{-1}>0$, $\eps > 0$. 
Then, there exists a dynamic algorithm that implicitly maintains $(\gamma, (1+\eps)\alpha)$-orientation of a dynamic graph $G$ with arboricity $\alpha$ as well as a $(\delta,\mu)$-refinement $H$ of $G$ wrt.\ this orientation such that $H$ is a forest. 
Insertion and deletion take amortised $O(\eps^{-6}\log^4{n})$ time.
\end{theorem}
In each copy, we can round $H$ and $G-H$ as in the proof of Theorem \ref{OutOrientationMain} in order to obtain an implicit out-orientation with update times independent of $\alpha$.

\subparagraph{Proof of Theorem \ref{thm:SW-Main}}
For the algorithm proving Theorem \ref{thm:SW-Main}, we store the following global data structure.
\begin{itemize}
    \item $Edges:$ Balanced binary search tree sorted by edge id. Initially empty.
\end{itemize}
Each vertex $u$ stores the following data structures:
\begin{itemize}
    \item $s(u)$, initialized to $0$.
    \item $OutNbrs_{u}$: Balanced binary search tree containing out-neighbours of $u$ sorted by vertex id. Initially empty
    \item $InNbrs_{u}$: Max heap containing in-neighbours indexed using $s_{v}$. Initially empty.
\end{itemize}
We update them using the operations in algorithms \ref{alg:basic-ops2}, \ref{alg:implicit-access2} and \ref{alg:insert-delete2}.
\begin{algorithm} 
    \caption{Basic operations}
    \label{alg:basic-ops2}
    \begin{algorithmic}
            \Function{add}{$e,u \rightarrow v$}
                \If{$e\in Edges$}
                    \State Update $|B_{e}^{u}| \leftarrow |B_{e}^{u}| + 1$
                    \State Update $|B_{e}| \leftarrow |B_{e}| + 1$
                \Else 
                    \State add $e$ to $Edges$ with keys 
                    \State $|B_{e}^{u}| = 1$, $|B_{e}^{v}| = 0$,    $|B_{e}| = 1$.
                \EndIf
                \If{$u \notin InNbrs_{v}$}
                    \State add $u$ to $InNbrs_{v}$ with key $s(u)$
                \EndIf
                \If{$v \notin OutNbrs_{u}$}
                    \State add $v$ to $OutNbrs_{u}$
                \EndIf
            \EndFunction
            \\
            \Function{remove}{$e,u \rightarrow v$}
                \If{$|B_{e}| = 1$}
                    \State remove $e$ from $Edges$
                \EndIf
                \If{$|B_{e}^{u}| = 1$}
                    \State Remove $u$ from $InNbrs_{v}$
                    \State Remove $v$ from $OutNbrs_{u}$
                \EndIf
            \EndFunction
            \\
            \Function{TightOutNbr}{$u$}
                \State $L \leftarrow$ \Call{NewNbrs}{$u$}
                \For{$w \in L$}
                    \If{$w\in OutNbrs_{u}$} \Comment{Note we always update $u$'s neighbours beforehand.}
                        \If{$s(w)-s(u) \leq -\eta/2$}
                            \State \Call{return}{} w
                        \EndIf
                    \EndIf
                \EndFor
                \For{$w \in \text{next} \frac{4s(v)}{\eta} OutNbrs_{u}$}
                    \If{$s(w)-s(u) \leq -\eta/2$}
                        \State \Call{return}{} w
                    \EndIf
                \EndFor
                \State \Call{return}{} null
            \EndFunction
            \\
            \Function{TightInNbr}{$u$}
                \State $w \leftarrow max(InNbrs_{u})$
                \If{$s(u)-s(w) \leq -\eta/2$}
                    \State \Call{return}{} w
                \EndIf
                \State \Call{return}{} null
            \EndFunction
            \\
            \Function{increment}{$u$}
            \State $s(u) \leftarrow s(u)+1$
            \For{$v \in \text{next} \frac{4s(v)}{\eta} OutNbrs_{u}$}
                \State Update $s_{u}\leftarrow s(u)$ in $InNbrs_{v}$
            \EndFor
            \EndFunction
            \\
            \Function{decrement}{$u$}
            \State $s(u) \leftarrow s(u)-1$
            \For{$v \in \text{next} \frac{4s(v)}{\eta} OutNbrs_{u}$}
                \State Update $s_{u} \leftarrow s(u)$ in $InNbrs_{v}$
            \EndFor
            \EndFunction
    \end{algorithmic}
\end{algorithm}
\begin{algorithm}
    \caption{Accessing the implicit orientation}
    \label{alg:implicit-access2}
        \begin{algorithmic}
            \Function{access}{$e,uv$}
                \State Update $|B_{e}^{u}|,|B_{e}^{v}| \leftarrow$ \Call{AccessLoad}{$e$} in $Edges$
            \EndFunction
            \\
            \Function{change}{$e,uv$}
                \State $x \leftarrow |B_{e}^{u}|$
                \State $y \leftarrow |B_{e}^{v}|$
                \State \Call{UpdateLoad}{$u,v,x,y$}
            \EndFunction
            \\
            \Function{reorient}{$e,u\rightarrow v$}
                \State \Call{access}{$e,uv$}
                \State \Call{remove}{$e,u\rightarrow v$}
                \State \Call{add}{$e,v\rightarrow u$}
                \State \Call{change}{$e,uv$}
            \EndFunction
            \\
            \Function{UpdateNbrs}{$u$}
                \State $L \leftarrow$ \Call{NewNbrs}{$u$}
                \For{$uv \in L$}
                    \State \Call{access}{$uv$}
                    \For{$i=1$ to 2}
                        \If{$B_{uv}^{u} = 0$}
                            \State Remove $u$ from $InNbrs_{v}$
                            \State Remove $v$ from $OutNbrs_{u}$
                        \Else 
                            \If{$u \notin InNbrs_{v}$}
                                \State add $u$ to $InNbrs_{v}$ with key $s_u \leftarrow s(u)$
                            \EndIf
                            \If{$v \notin OutNbrs_{u}$}
                                \State add $v$ to $OutNbrs_{u}$
                            \EndIf
                        \EndIf
                        \State $u,v \leftarrow v,u$
                    \EndFor
                \EndFor
            \EndFunction
        \end{algorithmic}
\end{algorithm}
\begin{algorithm}
    \caption{Inserting procedure: }
    \label{alg:insert-delete2}
        \begin{algorithmic}
            \Function{InsertCopy}{$e,uv$}
                \If{$s(u)\leq s(v)$}
                    \State \Call{add}{$e,u\rightarrow v$}
                    \State $w \leftarrow u$
                \Else 
                    \State \Call{add}{$e,v\rightarrow u$}
                    \State $w \leftarrow v$
                \EndIf
                \State \Call{UpdateNbrs}{$w$}
                \While{\Call{TightOutNbrs}{$w$}$\neq$ null}
                    \State $w' \leftarrow$ \Call{TightOutNbrs}{$w$}
                    \State \Call{reorient}{$ww',w\rightarrow w'$}
                    \State $w \leftarrow w'$
                    \State \Call{UpdateNbrs}{$w$}
                \EndWhile
                \State \Call{increment}{$w$}
            \EndFunction
            \\
            \Function{GammaInsert}{$e,uv$}
                \For{1 to $\gamma$}
                    \State \Call{InsertCopy}{$e,uv$}
                \EndFor
            \EndFunction
            \\
            \Function{DeleteCopy}{$e,u\rightarrow v$}
                \State \Call{remove}{$e,u\rightarrow v$}
                \State $w \leftarrow u$
                \State \Call{UpdateNbrs}{$w$}
                \While{\Call{TightInNbrs}{$w$}$\neq$ null}
                    \State $w' \leftarrow$ \Call{TightInNbrs}{$w$}
                    \State \Call{reorient}{$w'w,w'\rightarrow w$}
                    \State \Call{UpdateNbrs}{$w'$}
                \EndWhile
                \State \Call{decrement}{$w'$}
            \EndFunction
            \\
            \Function{GammaDelete}{$e,uv$}
                \For{1 to $\gamma$}
                    \If{$|B_{e}^{u}| >0 $}
                        \State \Call{DeleteCopy}{$e,u\rightarrow v$}
                    \Else 
                        \State \Call{DeleteCopy}{$e,v\rightarrow u$}
                    \EndIf
                \EndFor
            \EndFunction
        \end{algorithmic}
\end{algorithm}
Assuming we know $\rho_{max}$, $n$ and $\eps$, we set $\eta = 2\frac{\rho_{max}}{\gamma}$ and $\gamma = 128\log{(n)}/\eps^2$.

\subparagraph{Correctness} 
Correctness follows from arguments symmetric to those given in Lemma 3.7 in \cite{DBLP:conf/stoc/SawlaniW20}. We have to argue that the orientation at all times is locally $\eta$-stable. 
In order to do so, we observe that to break the local stability, we have to, at some point, increment the load of a vertex with a critically tight in-neighbour or decrement the load of a vertex with a critically tight out-neighbour.
We always check the edges to neighbours returned from the list, and the remaining neighbours are updated sufficiently often (see Lemma 3.7 in \cite{DBLP:conf/stoc/SawlaniW20}), and so this never happens.

\subparagraph{Analysis} 
We can insert a copy of an edge in time $O(\gamma(\gamma+|L|\log{n}))$. Indeed, by Remark \ref{rmk:chainlength} a maximal chain has length no more than $2\rho_{max}/\eta = O(\gamma)$. 
At each vertex in the chain we spend $O(\gamma)$ time checking for tight out-neighbours and $O(|L|\log{n})$ time retrieving and checking the lists. Similarly for delete. 
Hence, we arrive at Theorem \ref{thm:SW-Main}. 

\subparagraph{The scheduling algorithm}
Here we describe the ideas from \cite{DBLP:conf/stoc/SawlaniW20} needed for the algorithm in Theorem \ref{thm:schedule}:

We will maintain $O(\log n)$ copies of $G^{\gamma}$ some of which are only partially updated. 
In each copy, we maintain a different estimate $\rho_{est}$ of the maximum density. Based on this estimate we set $\rho^{max} = 2\rho_{est}$. 
We set $\gamma = 128\log{(n)}/\eps^{2}$, and in the $i$'th copy, we set $\rho_{est} = 2^{i-2}\gamma$. 
In copy $i$, we initialize the data structures from Theorem \ref{thm:SW-Main} and set $\eta_{i} = 2\frac{\rho^{max}_{i}}{\gamma}$. 
Furthermore, for each copy $G^{\gamma}_i$, we initialize an empty, sorted list of edges $pending_i$ using two BST (sorted by each endpoint of the edge).
Finally, we maintain a counter $active$, indicating that $G^{\gamma}_{active}$ is the currently active copy. Whenever an edge is inserted or deleted in $G$, we update as follows.

\subparagraph{Insertion:} To insert a copy of an edge $uv$, we insert it into all copies $i$ for which $\rho^{max}_{i}$ is greater than $\rho^{max}_{active}$. 
These insertions all take $\poly(\log{n},\eps^{-1})$ time by Theorem \ref{thm:SW-Main}. If the copy with $i = active +1$ satisfies $\hat{\rho}_{active+1} = \max_{v} s_{active+1}(v) \geq \rho^{max}_{active}$, we make $G^{\gamma}_{active+1}$ the active copy and increment $active$. 
Otherwise, we also insert $uv$ in the active copy.

For the remaining copies of $G^{\gamma}$, we cannot necessarily afford to insert the edge. We need a bound on the length of the maximal tight chains to achieve the update guarantees of Theorem \ref{thm:SW-Main}. By Remark \ref{rmk:chainlength} this is upper bounded by $\frac{2\cdot{}\max_{w} s(w)}{\eta}$ for a chain beginning or ending at $w$. 
Therefore, we only insert $uv$ in copy $i$ if at least one of $u$ or $v$ has load $s'$ strictly smaller than $\rho^{max}_{i}$.
We can always afford to do this, since this only increases the load of a vertex with load at most $s'$. 
Otherwise, we add the edge to $pending_i$. 

\subparagraph{Deletion:} To delete a copy of an edge $uv$, we delete it in all copies $i$ for which $\rho^{max}$ is greater than $\hat{\rho} = \max_{v\in G^{\gamma}_{active}} s(v)$. 
If the active copy satisfies $(1-\eps)\hat{\rho} < \rho^{est}_{active}$, we make $G_{active-1}$ the active copy. 
These deletions all take $\poly(\log{n},\eps^{-1})$ time by Theorem \ref{thm:SW-Main}.

For the remaining copies of $G^{\gamma}$, we do as follows:
If $uv \in pending_{i}$, we just remove it. 
Otherwise, delete a copy of $uv$ in $G^{\gamma}_i$, and locate the unique vertex $w$ whose loads was decremented. 
Now if $w$ is incident to an edge in $pending_{i}$, we insert one such edge into the copy and remove the corresponding edge from $pending_{i}$. We can continue doing this until the load of $w$ again becomes too high.

\subparagraph{Analysis:} We repeat the arguments of Sawlani and Wang \cite{DBLP:conf/stoc/SawlaniW20}.
Observe first that $ \rho_{est}^{active} <  \hat{\rho} \leq 2\rho_{est}^{active}$. Indeed, after each insertion, we check whether this is the case and update accordingly. Similarly for delete.

Secondly, $pending_{active}$ is always empty. Indeed, if $uv$ is put in $pending_{i}$ at some point during the algorithm, then at this point in time $s_i(v),s_i(u) = \rho^{max}_{i}$.
Hence, $G^{\gamma}_{i}$ can only become active if the load of either $u$ or $v$ is decreased, meaning that all edges incident to this vertex in $pending_{i}$ has been inserted. In particular, $uv$ is no longer in $pending_{i}$. Since we only insert an edge from $pending_i$ once, a copy of an edge is only inserted $O(\log{n})$-times in total. 
Sawlani and Wang state that we only have to remove one incident edge from $pending_{i}$ per decrement, which makes the algorithm a worst-case algorithm (see \cite{DBLP:conf/stoc/SawlaniW20} Theorem 1.1).

Now, we are ready to show Theorem \ref{thm:schedule}. We restate it for convenience:
\begin{theoremS}[Theorem \ref{thm:schedule}, implicit in \cite{DBLP:conf/stoc/SawlaniW20} as Theorem 1.1]
There exists a fully dynamic algorithm for scheduling updates that at all times maintains a pointer to a fully-updated copy with estimate $\rho_{est}$ where $(1-\eps)\rho_{est}/2\leq \alpha(G)<4\rho_{est}$. 
Furthermore, the updates are scheduled such that a copy $G'$ with estimate $\rho'$ satisfies $\alpha(G')\leq 4\rho'$.
The algorithm has amortised update times $O(\log^{4}(n)/\eps^{6})$.
\end{theoremS}
\begin{proof}
We use the algorithm from Theorem 1.1 in \cite{DBLP:conf/stoc/SawlaniW20} paraphrased above. 
This algorithm maintains $O(\log n)$ different (partial) copies of $G^{\gamma}$. 
For each copy, we maintain a (partial) copy of $G$.
Whenever this algorithm has removed all $\gamma$ copies of an edge from $pending_{i}$, we insert the edge into the corresponding copy of $G$.

Sawlani \& Wang show that the active copy is always fully updated and that $\hat{\rho} = \max_{v} s(v)$ satisfies that $\hat{\rho}/2 \leq \rho_{est}^{active} <  \hat{\rho}$.
Theorem \ref{thm:SW-EtaValid} (Corollary 3.2 in \cite{DBLP:conf/stoc/SawlaniW20}) then implies that $\rho(G)$ satisfies $(1-\eps)\hat{\rho} \leq \rho \leq \hat{\rho}$. 
Since $\rho(G) \leq \alpha(G) \leq 2\rho(G)$ when $n\geq 2$, we get the bounds for the active copy. 

For a non-active copy $G'$ with maximum density estimate $\rho'$, we need to show that $\alpha(G')\leq 4\rho'$. 
The algorithm only inserts an edge $uv$ in $G'^{\gamma}$ if one of $s(u)$ and $s(v)$ is strictly smaller than $2\rho'$.
Inserting this edge increases $s(w)$ by one for exactly one vertex $w$ at the end of a tight chain. 
Since Sawlani \& Wang reorient a tight chain incident to the vertex out of $u$ and $v$ with the lowest load, this never increases the load of any vertex above $2\rho'$. 
In particular, since $\rho(G') \leq \max_{w \in G'^{\gamma}} s(w) \leq 2\rho'$ by Theorem \ref{thm:SW-EtaValid} (Corollary 3.2 in \cite{DBLP:conf/stoc/SawlaniW20}), and since $\alpha(G') \leq 2\rho(G')$ for $n\geq 2$, we get the desired bound on the arboricity of the copy. 
\end{proof}

\section{Missing proofs from Section 3}
\label{apx:Refinements}
\begin{lemmaS}[Lemma \ref{lma:toptrees}, implicit in \cite{10.1145/1103963.1103966}] 
Let $F$ be a dynamic forest in which every edge $e = wz$ is assigned a pair of variables $X_{e}^{w}, X_{e}^{z} \in [0,1]$ s.t. $X_{e}^{w} + X_{e}^{z} = 1$. Then there exists a data structure supporting the following operations, all in $O(\log{|F|})$-time:
\begin{itemize}
    \item $link(u,v,X_{uv}^{u}, X_{uv}^{v})$: Add the edge $uv$ to $F$ and set $X_{uv}^{u}, X_{uv}^{v} = 1-X_{uv}^{u}$ as indicated.
    \item $cut(u,v)$: Remove the edge $uv$ from $F$.
    \item $connected(u,v)$: Return $\operatorname{true}$ if $u,v$ are in the same tree, and $\operatorname{false}$ otherwise.
    \item $add\_weight(u,v,x)$: For all edges $wz$ on the path $u\dots wz \dots v$ between $u$ and $v$ in $F$, set $X_{wz}^{w} = X_{wz}^{w}+x$ and $X_{wz}^{z} = X_{wz}^{z}-x$. 
    \item $min\_weight(u,v)$: Return the minimum $X_{wz}^{w}$ s.t.\ $wz$ is on the path $u\dots wz \dots v$ in $F$.
    \item $max\_weight(u,v))$: Return the maximum $X_{wz}^{w}$ s.t.\ $wz$ is on the path $u\dots wz \dots v$ in $F$.
\end{itemize}
\end{lemmaS}
\begin{proof}
\cite{10.1145/1103963.1103966} shows how to support the first 3 operations. 
The remaining operations also follow from \cite{10.1145/1103963.1103966}, but we give a proof for completeness:

Each path-cluster $C$ maintains the variables $min\_weight(C)$, $max\_weight(C)$, $extra(C)$ and $head(C) \in \partial C$. 
When a leaf-cluster $C$ with cluster path $\pi(C) = uv$ is created, we initialize $min\_weight(C) = max\_weight(C) = X_{uv}^{u}$, $extra(C) = 0$, $head(C) = v$. 
\begin{figure}[]%
    \centering
    \subfloat{{\includegraphics[width=8cm]{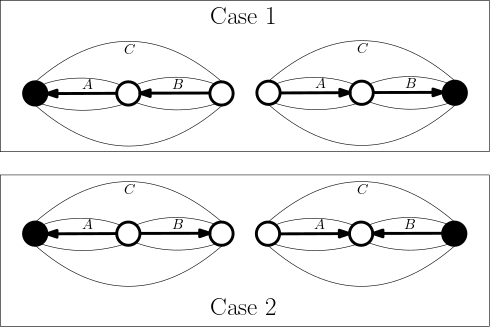} }}%
    \caption{Different choices of $head(A)$, $head(B)$ and $head(C)$. The head of $C$ is indicated by the filled vertex. The arrows point towards the heads of clusters $A$ and $B$.
    We update the information differently depending on how $head(B)$ relates to $head(C)$.}
    \label{fig:Toptree}%
\end{figure}
When a path-cluster $C = A \cup B$ with two cluster children is created, we set $head(C)$ to $head(A)$ if possible.
Otherwise, we set $head(C) = \partial C \cap \partial B$. 
If $head(C) = head(A)$ and $head(B) \in \partial A$ or if $head(C) = head(B)$ and $head(B) \notin \partial(A)$ (Case 1), we update:
\begin{align*}
    min\_weight(C) &= \min\{min\_weight(A),min\_weight(B) \} \\ 
    max\_weight(C) &= \min\{max\_weight(A),max\_weight(B) \}
\end{align*}
Otherwise (Case 2), we update:
\begin{align*}
    min\_weight(C) = \min\{min\_weight(A),1-max\_weight(B) \}\\ 
    max\_weight(C) = \min\{max\_weight(A),1-min\_weight(B) \}
\end{align*}
Finally, we set $extra(C) = 0$.

When a path-cluster $C = A \cup B$ with two cluster children is split, we set:
\begin{align*}
    extra(A) &= extra(A) + extra(C) \\
    min\_weight(A) &= min\_weight(A)+extra(C) \\ 
    max\_weight(A) &= max\_weight(A)+extra(C) 
\end{align*}
If $head(C) = head(A)$ and $head(B) \in \partial A$ or if $head(C) = head(B)$ and $head(B) \notin \partial A$ (Case 1), we set;
\begin{align*}
    extra(B) &= extra(B) + extra(C) \\
    min\_weight(B) &= min\_weight(B)+extra(C) \\ 
    max\_weight(B) &= max\_weight(B)+extra(C) 
\end{align*}
Otherwise (Case 2), we set:
\begin{align*}
    extra(B) &= extra(B) - extra(C) \\
    min\_weight(B) &= min\_weight(B)-extra(C) \\ 
    max\_weight(B) &= max\_weight(B)-extra(C) 
\end{align*}
Note that only Case 1 and Case 2 can occur due to the way we create the joined clusters.
To $add\_weight(u,v,x)$, we set $C = expose(u,v)$ and add $x$ to $min\_weight(C)$, $max\_weight(C)$ and $extra(C)$ if $head(C) = v$. Otherwise, we subtract $x$ from $min\_weight(C)$, $max\_weight(C)$ and $extra(C)$. 

To do $max\_weight(u,v)$, we again set $C = expose(u,v)$, and return $max\_weight(C)$ if $head(C) = v$ and $1-min\_weight(C)$ otherwise. $min\_weight(u,v)$ is similar.
\end{proof}

\section{Missing proofs of Section 4}
\label{apx:Forests}
\begin{lemmaS}[Lemma \ref{lma:split}, Implicit in \cite{henzinger2020explicit}]
Given black box access to an algorithm maintaining an $\alpha'$-bounded out-degree with update time $T(n)$, there exist an algorithm maintaining an $\alpha'$ pseudo-forest decomposition with update time $O(T(n))$.
\end{lemmaS}
\begin{proof}
We describe the algorithm below for completeness and stress that it is almost identical to one presented in \cite{henzinger2020explicit}. The algorithm runs the black box algorithm maintaining an $\alpha'$-bounded out-degree orientation.
Then, with constant overhead, it maintains $n$ pseudoforests $P_0, P_1, \dots , P_{n-1}$, but only the first $\alpha'$ pseudoforests are non-empty.
We will conceptually maintain an array $A_v$ of bits of length $n$ for each $v$, with $A_{v}(i)$ indicating the out-degree of $v$ in $F_{i}$.
The algorithm will preserve the following invariants for every $v \in V(G)$:
\begin{invariant}
\label{inv:pseudo1}
For all $i \in \{0, \dots,d^+(v)-1 \}$, $P_i$, will contain an out-edge of $v$.
\end{invariant}
\begin{invariant}
\label{inv:pseudo2}
For all $i \in \{0, \dots,n-1 \}$, we have that $A_v(i) = 1$ iff. $N^{+}(v) \cap P_i \neq \emptyset$.
\end{invariant}
Finally, we maintain an index $s_v$ for each $v$ indicating the smallest $i$ s.t. $A_v(i) = 0$, and an index $l_v$ indicating the largest $i$ s.t. $A_v(i) = 1$. 
Note that when the above invariants are maintained, we necessarily have $s_v > l_v$.

The algorithm updates in the following way: 
Suppose an edge $e$ is inserted/deleted from the graph. 
Then $e$ is inserted into the algorithm maintaining the out-orientation. 
Then three things might happen: $A)$ the orientation algorithm might reorient the orientation of an existing edge, $B)$ a new oriented edge might be inserted into the graph and finally $C)$ an oriented edge might be deleted.

Consider the case $1)$ firstly. Suppose $u \rightarrow v$ is reoriented, so that it now has direction $v \rightarrow u$. Then we update as follows:
\begin{enumerate}
    \item delete $u\rightarrow v$ from the forest where $e$ is currently residing. Set $A_u(i) = 0$ and update $s_u = i$ accordingly. 
    \item If $s_u \geq l_u$, all invariants are still maintained for $u$, and all we have to do is update $l_u = l_u-1$.
    If not, then we will move the out-edge of $u$, $u\rightarrow w$, residing in $F_{l_{u}}$ to $F_{s_u}$ 
    Then we update $s_u = l_u$ and $l_{u} = l_{u}-1$, and now all of $u$'s invariants are restored again.
    \item Insert $v \rightarrow u$ into $F_{l_{v}}$.
Then update $A(l_{v}+1) = 1$, $l_{v} = l_{v}+1$ and $s_{v} = s_{v}+1$, again restoring all of the invariants.
\end{enumerate}
In case $B)$, where we are inserting an edge, we perform only step $3$, and in case $C)$, where an edge is deleted, we perform steps $1$ and $2$.

\subparagraph{Analysis of the algorithm}
The update procedure clearly maintains the claimed invariants.
The following Lemma establishes correctness.
\begin{claim}[Implicit in \cite{henzinger2020explicit}]
After each update $P_0, P_1, \dots , P_{d-1}$ form a  pseudoarboricity decomposition of the graph $G$.
\end{claim}
\begin{proof}
Indeed by Invariant \ref{inv:pseudo1}, $E(G) \subset \bigcup \limits_{j=0}^{\alpha'-1} E(P_j)$ so the pseudoforests cover all of the edges of $G$.
Furthermore, every $v \in P_{i}$ has out-degree at most 1, and so any cycle in $P_i$ must be directed. 
Since by Invariant \ref{inv:pseudo1}, $|E(P_C)| \leq |P_C|$ for any connected component $P_C \subset P_{i}$, we must have that if $P_{C}$ is cyclic, then equality holds.
Deleting an edge $e$ from a cycle in such a cyclic and connected component ensures that $|E(P_{C}-e)| =  |P_{C}|-1$ i.e.\ $P_{C}-e$ is a tree. Thus $P_{C}$ is a pseudotree. 
\end{proof}
\end{proof}

\end{document}